\everydisplay{
	\abovedisplayskip=.22\baselineskip  plus.2ex minus.2ex\abovedisplayshortskip=-0.33\baselineskip plus.4ex minus.4ex
	\belowdisplayskip=.22\baselineskip plus.2ex minus.2ex\belowdisplayshortskip=0.37\baselineskip plus.2ex minus.2ex
}
\documentclass[journal,10pt]{IEEEtran}
\ifCLASSINFOpdf
\else
\fi

\usepackage[T1]{fontenc}
\usepackage[latin9]{inputenc}
\usepackage{array}
\usepackage{float}
\usepackage{units}
\usepackage{multirow}
\usepackage{amsmath}
\usepackage{amssymb}
\usepackage{graphicx}
\usepackage{subfigure}
\usepackage{esint}
\usepackage{xcolor}
\usepackage{epstopdf}
\usepackage{arydshln}
\usepackage{cases}
\usepackage[T1]{fontenc}
\usepackage[latin9]{inputenc}
\usepackage{array}
\usepackage{units}
\usepackage{multirow}
\usepackage{amsmath}
\usepackage{amssymb}
\usepackage{graphicx}
\usepackage{esint}
\usepackage{xcolor}
\usepackage{enumerate}
\usepackage{epstopdf}
\usepackage{arydshln}
\usepackage{caption}
\usepackage{threeparttable}
\usepackage{cases}
\usepackage{diagbox}
\usepackage{subfigure}
\usepackage{mathtools}
\allowdisplaybreaks[4]

\newcommand{\blockc}[1]{
	{0 \ \cdots \ 0}
}
\usepackage{float}
\usepackage[square, comma, sort&compress, numbers]{natbib}  
\makeatletter

\floatstyle{ruled}
\newfloat{algorithm}{tbp}{loa}
\providecommand{\algorithmname}{Algorithm}
\floatname{algorithm}{\protect\algorithmname}

\makeatother

\usepackage[english]{babel}

\makeatletter
\adddialect\l@ENGLISH\l@english
\makeatother

\begin{document}
	\linespread{0.9}
%
\title{\hspace*{-2.8mm} \LARGE \bf  
	Distributed Resilient Fixed-Time Control for Cooperative Output Regulation of MASs over Directed Graphs under DoS Attacks}
\author{Wenji~Cao, Lu~Liu,~\IEEEmembership{Senior Member,~IEEE,}  Dan~Zhang,~\IEEEmembership{Senior Member,~IEEE,} Gang~Feng,~\IEEEmembership{Fellow,~IEEE}
	\thanks{Wenji Cao, Lu Liu and Gang Feng are with the Department of Mechanical Engineering, City University of Hong Kong, Kowloon, Hong
		Kong (e-mail: wenjicao2-c@my.cityu.edu.hk; luliu@cityu.edu.hk; megfeng@cityu.edu.hk).}
\thanks{Dan Zhang is with the College of Information Engineering, Zhejiang University of Technology, Hangzhou, 310023, China, and also with the Zhejiang Key Laboratory of Intelligent Perception and Control for Complex Systems, Hangzhou 310023, China (danzhang@zjut.edu.cn).}
 \thanks{This work was supported in part by the Research
 	Grants Council of Hong Kong under Grants CityU-11205024 and 
 	CityU-11210424; in part by the National Nature Science Foundation
 	of China under Grants U2541213/62373314/62322315; and in part by the Fundament
 	Research for the Provincial Universities of Zhejiang under Grant RFC2024002. (\textit{Corresponding author: Gang~Feng})}
}
\maketitle


\begin{abstract}
This paper addresses the problem of fixed-time cooperative output regulation for linear multi-agent systems over directed graphs under denial-of-service attacks. 
A novel distributed resilient fixed-time controller is developed that comprises a distributed resilient fixed-time observer taking general directed graphs into consideration, and a distributed resilient fixed-time control law for each agent.
The proposed controller neither depends on Laplacian symmetry nor requires strong connectivity and detail-balanced condition, in contrast to existing distributed resilient fixed-time controllers.
Under the proposed controller, the regulated outputs converge to zero in a fixed time with its upper bound independent of the initial states of the multi-agent system. 
Ultimately, the efficacy of the proposed controller is demonstrated via a simulation example.
\end{abstract}



\IEEEpeerreviewmaketitle

\newtheorem{lemma}{Lemma}
\newtheorem{theorem}{Theorem}
\newtheorem{assumption}{Assumption}
\newtheorem{remark}{Remark}
\newtheorem{corollary}{Corollary}
\newtheorem{definition}{Definition}
\newtheorem{problem}{Problem}
\newtheorem{property}{Property}

\section{Introduction}\label{sec:1}
Over recent decades, cooperative control of multi-agent systems (MASs) has received significant research interest, owing to its vast range of applications, including power systems, multirobot systems, sensor networks \citep{jadbabaie2003coordination,nagata2002multi,ren2008distributed,ren2011distributed,yang2022overview}.
As one of its typical topics, cooperative output regulation (COR) has been well investigated.
Various distributed control methods have been developed to address the COR problem for linear MASs and nonlinear MASs
\cite{su2011cooperative,su2012general,ding2015348,wang2019robust,song2021distributed,zhou2022lyapunov}.

In real-world applications, communication networks of MASs are frequently susceptible to cyber-attacks. 
Denial-of-service (DoS) attacks, among the most prevalent types of cyber-attacks, disrupt the timely transfer of information by blocking or severing network connectivity \citep{de2015input}. 
There has been significant research on resilient cooperative control of MASs under DoS attacks recently, see, for example, \cite{lu2018distributed,11141750cheng,feng2016distributed,feng2019secure,8307749,10233067,deng2022resilient,deng2020distributed,zhang2022resilient,ZHANG2026112619,tang2020event}.
Particularly, the resilient consensus problem of homogeneous linear MASs is studied over both undirected graphs \citep{lu2018distributed,11141750cheng} and directed graphs
\citep{feng2016distributed,feng2019secure}.
Then, the resilient consensus problem of heterogeneous linear MASs over directed graphs is studied in \cite{8307749,10233067}. 
The problem of resilient COR is addressed for heterogeneous MASs over undirected graphs with linear agent dynamics \citep{deng2022resilient,deng2020distributed} and nonlinear agent dynamics \citep{zhang2022resilient}. The resilient formation control of nonlinear MASs is considered over both undircted graphs \citep{ZHANG2026112619} and directed graphs \citep{tang2020event}.
A critical observation is that the control protocols proposed in the literature mentioned earlier on resilient cooperative control of MASs ensure only the asymptotic or exponential convergence.

However, many practical engineering applications call for finite-time convergence rather than asymptotic or exponential convergence.
As a result, finite-time cooperative control of MASs has attracted significant attention in recent years \cite{wang2010finite,zhang2025finited,furchi2025finite,sarrafan2022resilient,cao2025resilient1,zhang2025finite}.
In particular, the problems of finite-time cooperative control for MASs are studied in \cite{wang2010finite,zhang2025finited,furchi2025finite} while the problems of resilient finite-time cooperative control for MASs subject to DoS attacks are considered in \cite{sarrafan2022resilient, cao2025resilient1,zhang2025finite}.
Nevertheless, there is one drawback for those finite-time cooperative control schemes, i.e., the upper bounds of {their} settling times rely on the initial states of the concerned systems.
Considerable efforts have been made to address this drawback via developing fixed-time cooperative control schemes \citep{liu2024fixed,zuo2015nonsingular,du2020distributed,10453626}.
For example, the fixed-time consensus problem of linear MASs is investigated over both undirected graphs \citep{liu2024fixed} and directed graphs \citep{zuo2015nonsingular}. 
In \cite{du2020distributed}, the same problem of nonlinear MASs is addressed over strongly connected and detail-balanced graphs.
The authors in \cite{10453626} propose a smooth distributed adaptive control scheme to address the fixed-time COR problem of nonlinear MASs over undirected graphs. 
More recently, the problems of resilient fixed-time cooperative control for MASs subject to DoS attacks are investigated \citep{yang2020observer,xu2024fixed,cao2025resilient2,cao2025resilient3}. 
Specifically, in \citep{yang2020observer}, the resilient fixed-time consensus problem is studied for nonlinear MASs over directed graphs subject to connectivity-maintained/broken attacks.  
The resilient fixed-time consensus problem of nonlinear MASs over undirected graphs is tackled by a distributed fixed-time fuzzy control strategy in \cite{xu2024fixed}. 
The resilient fixed-time COR problem of nonlinear MASs is investigated over undirected graphs in \cite{cao2025resilient2}, while the same problem of  heterogeneous linear MASs is considered over strongly connected and detail-balanced directed graphs in \cite{cao2025resilient3}.
However, to the best of our knowledge, the problem of resilient fixed-time cooperative control for MASs over general directed graphs subject to DoS attacks remains to be addressed. This serves as the motivation for our current work.

This paper studies the resilient fixed-time COR problem of heterogeneous linear MASs over general directed graphs under DoS attacks. 
The key challenge  in addressing this problem lies in how to develop a distributed controller that guarantees fixed-time convergence of the regulated outputs in the presence of both directed graphs and DoS attacks.
On one hand, it is noteworthy that the techniques in \cite{yang2020observer}, \cite{xu2024fixed} and \cite{cao2025resilient2}, which are developed for undirected graphs, cannot be adopted for directed graphs because they rely on the symmetric properties of Laplacian matrices. 
On the other hand, the method in \cite{cao2025resilient3} addresses directed graphs but requires the strong connectivity and detail-balanced condition, 
making it inapplicable to general directed graphs due to its reliance on the symmetric properties of the weighted variant of the Laplacian matrices. 
Thus, we develop some new techniques to tackle the COR problem of heterogeneous linear MASs over directed graphs under DoS attacks in this study.
The primary contributions of this paper are outlined below.

(i) A novel distributed resilient fixed-time observer is developed by taking general directed graphs and DoS attacks into consideration. 
This observer is constructed for each agent to estimate the state of the exosystem, and it is proven that the observer's state approaches the exosystem's state in a fixed time, even when the communication network of the MAS is a general directed graph and subject to DoS attacks. An explicit upper bound, independent of the initial states of both the observer and the exosystem, is also provided.
With the proposed observer, we then develop a novel distributed resilient fixed-time controller via the geometrical homogeneity principle \cite{bhat2005geometric}.
This controller ensures the fixed-time convergence to zero of the regulated outputs with its upper bound independent of the initial states of the MAS.

(ii) Compared with \cite{yang2020observer}, \cite{xu2024fixed}, \cite{cao2025resilient2}, and \cite{cao2025resilient3}, which consider either undirected graphs or strongly connected and detail-balanced directed graphs, this work addresses more general directed graphs. In particular, unlike \cite{cao2025resilient3}, where the controller design depends on the detail-balanced condition, this work does not need this restriction. 
Furthermore, we adopt a single non-symmetric Lyapunov function candidate for stability analysis in this work, in contrast to multiple symmetric Lyapunov function candidates required in \cite{cao2025resilient3}.
The proposed distributed controller is expected to have broader applications as general directed graphs often exist in practical MASs.

\section{Preliminaries and Problem Statement}
\subsection{Algebraic Graph Theory}
A directed graph $\mathcal{G}\!=\!(\mathcal V,\mathcal{E})$ describes the communication network among $N$ agents and the exosystem labeled by $0$, in which 
$\mathcal V=\{0,1,\cdots, N\}$ 
signifies the set of nodes and $\mathcal E\!=\!\{(j,i)\!:i,j\!\in\! \mathcal V\}$ indicates the set of edges.
$(j,i) \in \mathcal{E}$ indicates that node $i$ is a neighbor of node $j$, implying that node $i$ has access to the information from node $j$.
A directed path from node $j_1$ to node $j_k$ refers to  a succession of edges, represented as $(j_r, {j_{r+1}}),$ for $r=1,2,\cdots,k-1$. 
A directed graph  is said to contain a directed spanning tree if there is a root node (with no parent) so that all other nodes are reachable from it along directed paths.
The adjacency matrix $\mathcal A=[a_{ij}]\in \mathbb R^{(N+1)\times (N+1)}$  of graph $\mathcal G$ is defined as follows: $a_{ii}=0$, and for $i\neq j$, $a_{ij}>0$ if $(j,i)\in \mathcal E$, and $a_{ij}=0$ otherwise.  
Let $\mathcal V_a=\{1,2,\cdots, N\}$.
Denote $\mathcal H=[h_{ij}]\in \mathbb R^{N\times N}$, where  
$h_{ii}=\sum_{j={1}}^Na_{ij}+a_{i0}$ and $h_{ij}=-a_{ij}$, for $i\neq j$ and $i,j\in \mathcal V_a$.

\subsection{DoS Attacks}
In a MAS, agents interact with their neighbors to exchange information via a communication network. 
DoS attacks refer to situations in which information transfer between agents is denied by blocking or disrupting network connectivity. In other words, when a communication network is under DoS attacks, those agents with the corresponding edges under attack will not be able to send/receive information. In this work, we focus on the most severe form of DoS attacks, called the zero-topology attacks. 
Let $t_0$ be the initial time. 
For the $k$th DoS attack, $k \in \mathbb{N}_+$, we denote its starting instant by $t_k^a$ and its ending instant by $t_k$.
Over the interval $[t_0, t]$, let $\Pi_D(t_0,t)$ be the union of all time intervals during which DoS attacks occur, and let $\Pi_N(t_0,t)$ be the union of all time intervals free from attacks, that is, 
$$
\Pi_D(t_0,t) 
= \bigcup_{\{l:\, t_l^a \leq t\}}
\bigl[ t_l^a,\, \min\{t_l,t\} \bigr),
$$
and 
$$
	\Pi_N(t_0,t)=[t_0,t]\backslash \Pi_D(t_0,t),
$$
where $[t_0,t]\backslash \Pi_D(t_0,t)$ represents the relative complement of $\Pi_D(t_0,t)$ in $[t_0,t]$. The lengths of sets $\Pi_D(t_0,t)$ and $\Pi_N(t_0,t)$  are denoted as follows,
\begin{align*}
	\begin{array}{lll}
		|\Pi_D(t_0,t)|=\left\{
		\begin{array}{l}
			\sum_{l=1}^{k-1}(t_l-t_l^a),t\in [t_{k-1}, t^a_k),\\
			\sum_{l=1}^{k-1}(t_l-t_l^a)+(t-t_k^a),t\in [t_k^a,t_{k}),
		\end{array}\right.
	\end{array}
\end{align*}
and
\begin{align*}\label{dun}
	|\Pi_N(t_0,t)|=t-t_0-|\Pi_D(t_0,t)|.
\end{align*}

Taking DoS attacks into account, the corresponding communication graphs become time-varying.
Define $\mathcal A^{\vartheta(t)}=[a_{ij}^{\vartheta(t)}]$ as the adjacency matrix under DoS attacks, where $a_{ij}^{\vartheta(t)}=a_{ij}\cdot \vartheta(t)$ and $\vartheta(t)$ is specified as 
\begin{align*}
	\begin{array}{lll}
\vartheta(t)=\left\{
		\begin{array}{l}
		1,t\in [t_{k-1}, t^a_k),\\
			0,t\in [t_k^a,t_{k}).
		\end{array}\right.
	\end{array}
\end{align*}


\subsection{Notations and Definitions}
The following notations are adopted in this work.
Let $\mathbb R$ ($\mathbb R_+$) and $\mathbb N$ ($\mathbb N_+$) denote the sets of real (positive real) and natural (positive natural) numbers, respectively.
$\mathbb{R}^n$ and $\mathbb{R}^{n \times n}$ represent the spaces of $n$-dimensional real vectors and $n \times n$ real matrices, respectively.
For a vector $x\in \mathbb R^n$, denote $\|x\|_1=\sum_{i=1}^{n}|x_i|$, $\|x\|=(\sum_{i=1}^{n}x_i^2)^{\frac{1}{2}}$, and $\text{sig}^c(x)=[\text{sign}(x_1)|x_1|^c, \text{sign}(x_2)|x_2|^c,\cdots, \text{sign}(x_n)|x_n|^c]^T$, where $\text{sign}(x_i)$ indicates the sign function of $x_i$ and $c\in \mathbb R_+$.
The notation $\otimes$ is the Kronecker operator. 
${I}_n\in \mathbb R^{n\times n}$ represents the identity matrix.
$\textbf{0}_{n}\in \mathbb R^n$ denotes the vector with all zero entries.
Given a matrix $A\in \mathbb R^{n\times n}$, let its minimum eigenvalue and maximum eigenvalue be represented as $\lambda_m(A)$ and $\lambda_M(A)$, respectively, and
$A > 0$ ($\geq 0$) indicates that $A$ is positive definite (positive semi-definite).

\begin{definition}\label{FTstable} \citep{bhat2000finite,polyakov2011nonlinear}
Consider the system
	\begin{equation}\label{finiteq}
		\dot{x}(t) = g(x,t),
	\end{equation}
	where $x(t) \in \mathbb{R}^n$ is its state, $g(x,t): \mathbb{R}^n\times \mathbb R_+ \to \mathbb{R}^n$ denotes a nonlinear function that may exhibit discontinuities,  and $g(0,t) = 0$. Solutions of \eqref{finiteq} are interpreted in the context of Filippov's framework \citep{filippov2013differential}.
	Let $x(0) = x_0$ and denote the solution as $x(t, x_0)$. The origin is said to be:
	\begin{enumerate}[(1)]
	\item \emph{Globally finite-time stable} if there is a function $\bar{T}(x_0): \mathbb{R}^n \setminus \{0\} \to \mathbb{R}_+$ such that:
		\begin{enumerate}[(a)]
			\item  For each $x_0 \in \mathbb{R}^n \setminus \{0\}$, 
			$x(t, x_0)$ is defined on $[0, \bar{T}(x_0))$, $x(t, x_0) \in \mathbb{R}^n \setminus \{0\}$, $\forall t \in [0, \bar{T}(x_0))$, and 
			$\lim_{t \to \bar{T}(x_0)} x(t, x_0) = 0$. (\emph{Global finite-time convergence})
			\item  For each open neighborhood $\mathcal{W}_x$ of the origin, there is an open neighborhood $\mathcal{W}_0 \subset \mathbb{R}^n$ of the origin such that for every $x_0 \in \mathcal{W}_0\setminus\{0\}$, $x(t, x_0) \in \mathcal{W}_x$ holds, $\forall t \in [0, \bar{T}(x_0))$. (\emph{Lyapunov stability})
		\end{enumerate}
		\item \emph{Globally fixed-time stable} if it is globally finite-time stable and $\bar{T}(x_0)$ is globally bounded by a constant $T \in \mathbb{R}_+$, independent of $x_0$.
	\end{enumerate}
\end{definition}

\subsection{Problem Statement}
Consider a MAS composed of an exosystem and $N$ agents.
 The dynamics of agent $i$ are given by
\begin{align}\label{agent}
	\begin{array}{lll}
		\dot x_i(t)=A_ix_i(t)+B_iu_i(t)+E_iv(t),\\
		e_i(t)=C_ix_i(t)+F_iv(t),
	\end{array}
\end{align}where $x_i(t)\in\mathbb R^{n_i}$ and $u_i(t)\in\mathbb R^{m_i}$ are its state and input, respectively, $e_i(t)\in\mathbb R^{p_i}$ is its regulated output, $n_i$ and $m_i$ are two positive integers with $m_i\leq n_i$, 
$A_i$, $B_i$, $C_i$, $E_i$ and $F_i$ are real matrices, 
$(A_i,B_i)$ is controllable with $B_i$ {being of full column rank}; $v(t)\in \mathbb R^q$ denotes the reference input to be
followed and/or the disturbance to be rejected, which is generated by the following exosystem,
\begin{align}\label{exo}
	\dot v(t)=Sv(t),
\end{align}
where $S\in \mathbb R^{q\times q}$.

This work aims to tackle the resilient fixed-time COR problem, formally defined below:

\textbf{Resilient Fixed-Time COR Problem:}
Given the MAS (\ref{agent})--(\ref{exo}) over directed graphs under DoS attacks, design a distributed controller such that
the resulting closed-loop system is fixed-time stable at the origin
when $v(t)\equiv0$; and 
the regulated output of each agent approaches zero in a fixed time, i.e., 
\begin{equation}\label{qe4}
	e_i(t)=0, \forall t\geq T,
\end{equation}
where $T\in \mathbb R_+$ is independent of the initial states.

To tackle this problem, we introduce some assumptions below.

\begin{assumption}\label{assdos}
The duration of  DoS attacks satisfies
	\begin{align*}\begin{array}{lll}
			|\Pi_D(t_0,t)| \leq \nu_d+\frac{t-t_0}{p_d}, \forall t \geq t_0,
		\end{array}
	\end{align*}
where $\nu_d>0$ and $p_d >1$.
\end{assumption}
\begin{assumption}\label{ass1}
For directed graph $\mathcal{G}$, there exists a spanning tree with node $0$ as its root.
\end{assumption}
\begin{assumption}\label{reglator}
	The regulator equations
	\begin{subequations}\label{req}
		\begin{align}
			A_i\Pi_i+B_i\Gamma_i+E_i&=\Pi_iS,\label{Za}\\
			C_i\Pi_i+F_i&=0, i\in \mathcal V_a,\label{Zb}
		\end{align}
	\end{subequations}have solutions $(\Pi_i,\Gamma_i)$.
\end{assumption}
\begin{remark}
	Assumption \ref{assdos} ensures that DoS attack duration remains bounded by a proportion of the overall time interval \citep{de2015input,deng2022resilient,zhang2022resilient}. 
	Assumptions \ref{ass1}--\ref{reglator} are made as in \cite{su2011cooperative,deng2022resilient,deng2020distributed,cao2025resilient3}, which are necessary and standard in solving the problem of COR for MASs. 
\end{remark}

Before proceeding, we present the following lemma, which describes essential properties of matrix $\mathcal H$ related to the directed graph $\mathcal G$.
\begin{lemma}\citep{mei2012distributed,dong2022110571}\label{lem1}
Let $\mathcal H$ be the matrix corresponding to the graph $\mathcal G$ under Assumption \ref{ass1}. The following statements are true:\\
	(i) Each eigenvalue of $\mathcal H$ has a positive real part.\\
	(ii) There is a matrix
	$K\triangleq\text{diag}\{k_1,k_2,\cdots,k_N\}>0$
	such that
	$\mathcal H^T K+K \mathcal H-2I_N\geq0$.
\end{lemma}

\section{Main Results}
This section first presents a novel distributed resilient fixed-time observer that accounts for directed graphs and DoS attacks. 
Subsequently, a distributed resilient fixed-time controller is developed, and the resulting closed-loop system is analyzed.
\subsection{Distributed Resilient Fixed-Time Observer}
Taking directed graphs and DoS attacks into account, the distributed resilient fixed-time observer with state $\eta_i(t)$ is  constructed as follows,
\begin{align}	\label{ob}
\dot {\eta}_i(t)=&S\eta_i(t)-\mu_1\hat \varsigma_i(t)-\mu_2\text{sig}^\alpha( \hat\varsigma_i(t) )\nonumber\\
&-\mu_3\text{sig}^\beta( \hat\varsigma_i(t) ), i\in \mathcal V_a,
\end{align}where 
$\hat \varsigma_i(t)=\sum_{j=0}^N a_{ij}^{\vartheta(t)}(\eta_i(t)-\eta_j(t))$ with $\eta_0(t)=v(t)$, $\alpha$, $\beta$, $\mu_1$, $\mu_2$ and $\mu_3$ are positive constants to be specified.


Define the estimation error as $\tilde \eta_i(t)=\eta_i(t)-v(t)$.
Let $\varsigma_i(t)=\sum_{j=0}^Na_{ij}(\eta_i(t)-\eta_j(t))$  and $\Phi_i(t) = \mu_1 \varsigma_i(t) + \mu_2 \text{sig}^\alpha(\varsigma_i(t)) + \mu_3 \text{sig}^\beta(\varsigma_i(t))$. The corresponding stacked vectors are denoted as $\tilde{\eta}(t) = [\tilde{\eta}_1^T(t), \tilde{\eta}_2^T(t), \cdots, \tilde{\eta}_N^T(t)]^T$, $\varsigma(t) = [\varsigma_1^T(t), \varsigma_2^T(t), \cdots, \varsigma_N^T(t)]^T$, and $\Phi(t) = [\Phi_1^T(t), \Phi_2^T(t), $ $\cdots, \Phi_N^T(t)]^T$.
Note that $\varsigma(t)=(\mathcal H\otimes I_q)\tilde \eta(t)$.
Since the most severe form of DoS attacks is considered in this work, it then follows from (\ref{ob}) that
\begin{subnumcases}
	{\label{obq1}	\dot \varsigma(t)=}
	(I_N\otimes S)\varsigma(t)-(\mathcal H\otimes I_q)\Phi(t),
	t\in \Pi_N(t_0,t),\label{obqq11}\\
	(I_N\otimes S)\varsigma(t),t\in \Pi_D(t_0,t).\label{obqq12}
\end{subnumcases}

Next, the following technical result will be stated.
\begin{theorem}\label{lemmat}
	Consider the exosystem (\ref{exo}) and the distributed observer (\ref{ob}) over directed graphs under DoS attacks. Suppose that Assumptions \ref{assdos}--\ref{ass1} hold. Given the DoS attack parameters $\nu_d$ and $p_d$, if there are positive constants $\alpha$, $\beta$, $\mu_1$, $\mu_2$ and $\mu_3$ such that $0<\alpha<1<\frac{1}{\alpha}<\beta$, $\mu_1>\|(K\otimes S)\|$, and the following conditions hold:\\
    (i) $\tfrac{c_1}{5c_2}(p_d-1)-c_5>0$,\\
	(ii) ${\hat c_2}e^{\hat c_2\nu_d}-{\hat c_1(p_d-1)}<0$,\\
	(iii) $e^{\hat c_2(\tfrac{t_o-\bar t_o}{p_d}+\nu_d)}-\tfrac{\hat c_1(p_d-1)}{p_d}(t_o-\bar t_o)+\hat c_1\nu_d\leq 0$,\\ 
	then the estimation error approaches zero in a fixed time $t_o$, i.e.,
	$$\tilde \eta_i(t)=0, \forall t\geq t_o,$$
	where $K$ is defined in Lemma \ref{lem1}, $\bar t_o\in \mathbb R_+$ and $t_o\in \mathbb R_+$ are calculated by 
	\begin{align}\label{bto}
c_2e^{(\tilde c_1(p_d-1)-\tilde c_2)\frac{(t-t_0)}{p_d}-\tilde c_1\nu_d-\tilde c_2\nu_d}-c_2-c_3=0,
\end{align}
	and 
	\begin{align}\label{to}
\hat c_2e^{\hat c_2(\frac{t-\bar t_o}{p_d}+\nu_d)}-\hat c_1(p_d-1)=0,
	\end{align}respectively, 
$\hat c_1=\tfrac{c_1(1-\alpha)}{5c_4(\alpha+1)}$,
$\hat c_2=\tfrac{c_5(1-\alpha)}{\alpha+1}$, 
$\tilde c_1=\tfrac{c_1(\beta-1)}{5 c_2(\beta+1)}$,
$\tilde c_2=\tfrac{c_5(\beta-1)}{\beta+1}$,
$c_1=\tfrac{1}{2}\min\{\mu_1^2-\|(K\otimes S)\|^2,\mu_2^2,\mu_3^2(Nq)^{1-\beta}\}$, \\
$c_2=\max\{\tfrac{\mu_1k_M}{2},\tfrac{\mu_2k_M}{\alpha+1}(Nq)^{\frac{1-\alpha}{2}},\tfrac{\mu_3k_M}{\beta+1}\}$,\\
$c_3=(3Nq)^{\frac{\beta-1}{\beta+1}}(\max\{\tfrac{\mu_1k_M}{2},\tfrac{\mu_2k_M}{\alpha+1},\tfrac{\mu_3k_M}{\beta+1}\})^{\frac{2\beta}{\beta+1}}$,\\
$c_4=\max\{\!(\!\tfrac{\mu_1k_M}{2}\!)\!^{\frac{2\alpha}{\alpha+1}}\!(Nq)\!^{\frac{1-\alpha}{\alpha+1}}, (\!\tfrac{\mu_2k_M}{\alpha+1}\!)\!^{\frac{2\alpha}{\alpha+1}}\!(Nq)^{1-\alpha}, (\!\tfrac{\mu_3k_M}{\beta+1}\!)\!^{\frac{2\alpha}{\alpha+1}}\}$,
$c_5=\max\{q^{\frac{1-\alpha}{2}}, q^{\frac{\beta-1}{2}}\}\|S\|(\beta+1)$, 
$k_M=\max_{i=1}^N\{k_i\}$.

Moreover, the upper bound of the settling time is
	given by 
	\begin{align}
		t_o=&t_0+\frac{p_d(\ln(1+\frac{c_3}{c_2})+(\tilde c_1+\tilde c_2)\nu_d)}{\tilde c_1(p_d-1)-\tilde c_2}\nonumber\\
		&+\frac{p_d}{\hat c_2}(\ln(\frac{\hat c_1(p_d-1)}{\hat c_2})-\hat c_2\nu_d).
	\end{align}
\end{theorem}
\textbf{\textit{Proof:}}
Refer to Appendix B.
\begin{remark}
The existence of a solution to (\ref{bto}) is ensured by condition (i), while that of (\ref{to}) is ensured by condition (ii). Condition (iii) guarantees that the minimum of $g_2(t)$, defined in (\ref{g2}), is non-positive.
	From conditions (i)--(iii), the observer gains $\mu_1$, $\mu_2$ and $\mu_3$ should be selected such that
\begin{align}
	\frac{c_1}{c_4}
		&\geq \frac{5c_5}{p_d-1}\max\{e^{\hat c_2\nu_d},e^{(1+\frac{\hat c_2p_d\nu_d}{p_d-1})}\},\\
		\frac{c_1}{c_2}&\geq \frac{5c_5}{p_d-1}.
	\end{align}
\end{remark}
\begin{remark}
Conditions (i)--(iii) in Theorem \ref{lemmat} indicate that increasing the observer gains $\mu_1$, $\mu_2$ and $\mu_3$ can resist stronger DoS attacks, which are characterized by a decrease of $p_d$ and/or an increase of $\nu_d$.
\end{remark}
\subsection{Distributed Resilient Fixed-Time Controller}
With the distributed observer (\ref{ob}) and the geometric homogeneity principle \citep{bhat2005geometric}, this subsection presents a distributed resilient fixed-time controller.

Let $\tilde x_i(t)=x_i(t)-\Pi_i\eta_i(t)$, where $\Pi_i$ is given by (\ref{req}). 
Given that $(A_i,B_i)$ is controllable and $B_i$ is of full column rank, it follows from Lemma \ref{vector} (i) in Appendix \ref{lemmas} that there is an output $\varrho_i(t)=R_i\tilde x_i(t)=[\varrho_{i1}(t), \varrho_{i2}(t), \cdots, \varrho_{im_i}(t)]^T\in \mathbb R^{m_i}$ which has complete vector relative degree $(q_{i1}, q_{i2}, \cdots, q_{i{m_i}})$.
Then, we propose the distributed resilient fixed-time control law for agent $i$ as follows,
	\begin{align}\label{controller1n}
		u_i(t)=&X_i^{-1}(\omega_i(t)-U_i\tilde x_i(t))+\Gamma_i\eta_i(t), i\in \mathcal V_a,
	\end{align}where $\eta_i(t)\in \mathbb R^q$ is the state of the observer (\ref{ob}), 
$X_i=
\begin{bmatrix}
	\vspace{-0.5em}R_{i1}A_i^{q_{i1}-1}B_i\\
	\vdots\\
	R_{im_i}A_i^{q_{i{m_i}}-1}B_i
\end{bmatrix}
\in \mathbb R^{m_i\times m_i}$, 
$U_i=
\begin{bmatrix}
	\vspace{-0.5em}R_{i1}A_i^{q_{i1}}\\
	\vdots\\
	R_{im_i}A_i^{q_{i{m_i}}}
\end{bmatrix}
\in \mathbb R^{m_i\times n_i}$,
$R_{ir}$ is the $r$-th row of $R_i$, 
$\Gamma_i$ is defined in (\ref{req}), 
$\omega_i(t)=[\omega_{i1}(t),\omega_{i2}(t),\cdots,\omega_{im_i}(t)]^T$, $\omega_{ir}(t)=-\sum_{k=1}^{q_{ir}}(\psi_{{ir}}^k\text{sig}^{\gamma^k_{ir}}(\varrho_{ir}^{(k-1)}(t))+\bar \psi_{{ir}}^k\text{sig}^{\bar{\gamma}^k_{ir}}( \varrho_{ir}^{(k-1)}(t))),$
$r=1,\cdots,m_i$,  $\gamma_{ir}^k$, 
$\bar{\gamma}_{ir}^k$, $\psi_{ir}^k$ and $\bar \psi_{ir}^k$ are selected by Lemma \ref{lemmafinite} in Appendix \ref{lemmas}.


We now proceed to give the main result of our work.
\begin{theorem}\label{main theorem}
	Consider the heterogeneous linear MAS (\ref{agent})--(\ref{exo}) under Assumptions \ref{assdos}-\ref{reglator}.
	If conditions (i)--(iii) in Theorem \ref{lemmat} are satisfied, the resilient fixed-time COR problem is solved under the distributed resilient fixed-time controller  (\ref{controller1n}) together with (\ref{ob}).
Furthermore, the upper bound of the settling time for the resulting closed-loop system is given by $t_{a}\triangleq t_o+t_c$, $t_o$ is given by Theorem \ref{lemmat},  $t_{c}=\max_{ir}\{t_{{c}_{ir}}\}$, with
	\begin{eqnarray}\label{t_cij}
	t_{{c}_{ir}}\!=\!\frac{\gamma_{ir}\lambda_{M}(P_{ir})\lambda_{M}^{\frac{1-\gamma_{ir}}{\gamma_{ir}}}\!\!(P_{ir})}{(1-\gamma_{ir})\lambda_{m}(Q_{ir})}+\frac{\bar{\gamma}_{ir}\lambda_{M}(\bar P_{ir})\lambda_{M}^{\frac{\bar{\gamma}_{ir}-1}{\bar{\gamma}_{ir}}}\!\!(\bar P_{ir})}{(\bar{\gamma}_{ir}-1)\lambda_{m}(\bar Q_{ir})},
\end{eqnarray}
{where} $P_{ir}>0$ and $\bar P_{ir}>0$ satisfy the Lyapunov equations,
\begin{align}
	P_{ir}\Psi_{ir}+{\Psi_{ir}^T}P_{ir} &= -Q_{ir}, \label{LY} \\
	\bar P_{ir}\bar \Psi_{ir}+{\bar \Psi_{ir}}^T\bar P_{ir} &= -\bar Q_{ir}, \label{LY2}
\end{align}
	for any $Q_{ir}>0$ and $\bar Q_{ir}>0$, respectively, $\Psi_{ir}$ and $\bar \Psi_{ir}$ are given as follows,
	\begin{align}\label{Aij1}
		\Psi_{ir}=\left[\begin{matrix}
				\renewcommand{\arraystretch}{0.7}
			0&1&\cdots&0\\
			\vdots&\vdots&\ddots&\vdots\\
			0&0&\cdots&1\\
			-\psi_{ir}^1&-\psi_{ir}^2&\cdots&-\psi_{ir}^{q_{ir}}
		\end{matrix}\right],
	\end{align}
	\begin{align}\label{Aij2}
		\bar \Psi_{ir}=\left[\begin{matrix}
				\renewcommand{\arraystretch}{0.7}
			0&1&\cdots&0\\
			\vdots&\vdots&\ddots&\vdots\\
			0&0&\cdots&1\\
			-\bar \psi_{ir}^1&-\bar \psi_{ir}^2&\cdots&-\bar \psi_{ir}^{q_{ir}}
		\end{matrix}\right].\end{align}
\end{theorem}

\textbf{\textit{Proof:}}
Recall that $\tilde \eta_i(t)=\eta_i(t)-v(t)$ and $\tilde x_i(t)=x_i(t)-\Pi_i\eta_i(t)$. From (\ref{Zb}), one has
\begin{align}\label{ei}
	e_i(t)=C_i\tilde x_i(t)+C_i\Pi\tilde \eta_i(t).
\end{align}According to Theorem \ref{lemmat}, $\tilde \eta_i(t)=0$ holds, for all $t\geq t_o$ and any initial state $\tilde \eta_i(0)\in \mathbb R^q$. 

Then, we will show the convergence of $\tilde x_i(t)$ in a fixed time. Recall that $\Phi_i(t) = \mu_1 \varsigma_i(t) + \mu_2 \text{sig}^\alpha(\varsigma_i(t)) + \mu_3 \text{sig}^\beta(\varsigma_i(t))$.
From Theorem \ref{lemmat}, we can directly conclude that $\Phi_i(t)=0$, for all $t\geq t_o$. 
It then follows from 
(\ref{agent}), (\ref{Za}) and (\ref{controller1n}) that 
\begin{align}\label{refq2}
	\dot{\tilde x}_i(t) &=A_i\tilde x_i(t)+B_i\tilde u_i(t), \forall t\geq t_o,
\end{align}where $\tilde u_i(t)=X_i^{-1}(\omega_i(t)-U_i\tilde x_i(t)).$
Recall that $(A_i,B_i)$ is controllable and $B_i$ is of full column rank.
According to Lemma \ref{vector} (ii) in Appendix \ref{lemmas}, system (\ref{refq2}) can be decomposed via the nonsingular transformations  $\bar u_i(t)=G_i\tilde u_i(t)$ and $\bar x_i(t)=T_i\tilde x_i(t)$ as follows,
\begin{align}\label{barx}
	\dot{\bar x}_i(t)=\bar A_i \bar x_i(t)+\bar B_i\bar u_i(t),
\end{align}where 
$\bar A_i=[\bar A_{i_{lr}}],
\bar B_i=\text{blockdiag}\{\Theta_{q_{i1}}, \Theta_{q_{i2}}, \cdots,\Theta_{q_{i{m_i}}}\}$,
$\bar A_{i_{ll}}=\left[\begin{matrix}
	\textbf 0_{q_{il}-1}& I_{q_{il}-1}\\
	\zeta_{i_{ll}}^1&\bar \zeta_{i_{ll}}\\
\end{matrix}\right]$,
$\bar A_{i_{lr}}=\left[\begin{matrix}
	\textbf 0_{(q_{ir}-1)\times q_{ir}}\\
	\zeta_{i_{lr}}\\
\end{matrix}\right]$,
$\bar \zeta_{i_{ll}}=[\zeta_{i_{ll}}^2, $ $\cdots, \zeta_{i_{ll}}^{q_{il}}]$,
$\zeta_{i_{lr}}=[\zeta_{i_{lr}}^1, \zeta_{i_{lr}}^2, \!\cdots,  \zeta_{i_{lr}}^{q_{ir}}]$,
$\Theta_{q_{il}}=[0, \cdots, $ $ 0, 1]^T\!\in\! \mathbb R^{q_{il}}$, $\sum_{j=1}^{{m_i}}q_{ir}={n_i}$, $l\!\neq\! r$ and $l,r=1,$ $2, \cdots,{m_i}$. 

Then, by denoting the first component of the $r$-th block of $\bar x_i(t)$ as $\varrho_{ir}(t)$, the output $\varrho_i(t)=R_i\tilde x_i(t)$ is constructed, where $R_i=\text{blockdiag}\{D_{q_{i1}},D_{q_{i2}}, \cdots,D_{q_{i{m_i}}}\}T_i\in \mathbb R^{m_i\times n_i}$ with ${D_{q_{ir}}}=[1, 0, \cdots, 0]$, $r=1,2,\cdots,m_i$.
It follows from  Lemma \ref{vector} (i) in Appendix \ref{lemmas} that the output $\varrho_i(t)$ has complete vector relative degree $(q_{i1},q_{i2},\cdots, q_{im_i})$.
By letting $\bar\varrho_{ir}(t)= [\varrho_{ir}(t), \varrho_{ir}^{(1)}(t), \cdots, \varrho_{ir}^{(q_{ir}-1)}(t)]^T$, it then follows from (\ref{refq2}) and Lemma \ref{vector} (i) in Appendix \ref{lemmas} that
\begin{align}\label{subsystem1}
	\dot{\bar\varrho}_{ir}(t)=A_{\varrho_{ir}}\bar\varrho_{ir}(t)+B_{\varrho_{ir}}\omega_{ir}(t), r=1,2,\cdots,m_i,
\end{align}
where $A_{\varrho_{ir}}=\left[\begin{matrix}
	\textbf {0}_{q_{ir}-1}& {I}_{q_{ir}-1}\\
	0&\textbf 0_{q_{ir}-1}^T
\end{matrix}	\right]$ 
and $B_{\varrho_{ir}}=\left[
\begin{matrix}
	\textbf {0}_{q_{ir}-1}\\
	1\\
\end{matrix}\right]$.
It follows from Lemma \ref{lemmafinite} in Appendix \ref{lemmas} that systems (\ref{subsystem1}) with control laws $\omega_{ir}(t)$, $r=1,2,\cdots,m_i$ are globally fixed-time stable at the origin with an upper bound $t_c=\max_{ir}\{t_{c_{ir}}\}$, in which $t_{c_{ir}}$ is defined in (\ref{t_cij}). 
Recall that $\bar x_i(t)=[\bar\varrho_{i1}^T(t), \bar\varrho_{i2}^T(t), \cdots, \bar\varrho_{im_i}^T(t)]^T$, 
{$\bar x_i(t)=T_i\tilde x_i(t)$} with $T_i$ being nonsingular. Then, one can infer that  $\tilde x_i(t)=0$, for all $t\geq t_a$ and any initial state $\tilde x_i(0)\in \mathbb R^{n_i}$, where $t_a=t_o+t_c$.
	Therefore, one has $e_i(t)=0,$ for all $t\geq t_a$.
The proof is hereby completed.
	\section{A Simulation Example}
Consider a MAS composed of an exosystem and five agents, with its corresponding communication graph illustrated in Figure \ref{fig1}. The adjacency weights are set to $a_{ij}=1$ if $(j,i)\in \mathcal E$, and $a_{ij}=0$ otherwise.
One can observe that Assumption \ref{ass1} is satisfied.
Agent $i$ is modeled as an inverted pendulum on a cart \cite{wang2019robust}.
The linearized model of agent $i$ is given by (\ref{agent}), where 
\begin{align*}A_i=\left[
	\begin{matrix}
		\renewcommand{\arraystretch}{0.7}
			\vspace{-0.2em}
		0&1&0&0\\
			\vspace{-0.2em}
		0&0&g&0\\
			\vspace{-0.2em}
		0&0&0&1\\
			\vspace{-0.2em}
		0&\frac{f_i}{l_{i} M_{1i}}&\frac{(M_{1i}+M_{2i})g}{l_{i}M_{1i}}&-\frac{f_i}{M_{1i}}
	\end{matrix}\right],B_i=\left[
	\begin{matrix}
		\renewcommand{\arraystretch}{0.7}
			\vspace{-0.2em}
		0\\
			\vspace{-0.2em}
		0\\
			\vspace{-0.2em}
		0\\
			\vspace{-0.2em}
		\frac{1}{l_{i}M_{1i}}
	\end{matrix}\right],\\
	E_i=\left[
	\begin{matrix}
		\renewcommand{\arraystretch}{0.7}
			\vspace{-0.2em}
		0&0\\
			\vspace{-0.1em}
		\frac{\chi_{i1}+\chi_{i2}}{M_{1i}}&0\\
			\vspace{-0.2em}
		0&0\\
			\vspace{-0.1em}
		\frac{\chi_{i2}}{l_{i} M_{1i}}&0
	\end{matrix}\right],
	C_i=\left[
	\begin{matrix}
		1&0&-l_{i}&0
	\end{matrix}\right], F_i=\left[
	\begin{matrix}
		1&2
	\end{matrix}\right].\end{align*}
Table I presents the parameters of the agents.
The system matrix of exosystem (\ref{exo}) is given as $
	S=
\begin{bmatrix}
		0&-0.2\\
			\vspace{-0.2em}
		0.2&0\\
	\end{bmatrix}.
$
One can verify that Assumption \ref{reglator} is satisfied.


\begin{small}
	\begin{figure}[http]
		\centering
		\captionof{table}{{The parameters of the $i$th agent.}}
		\scalebox{0.85}{
			\begin{tabular}{|l|l|}
				\hline
				Parameters&Meaning\\
				\hline
				$M_{1i}=2\cdot i$ kg& Mass of cart\\ 	
				$M_{2i}=0.5\cdot i$ kg& Mass of pendulum\\ 
				$l_i=1\cdot i $ m & Length of pendulum\\
				$g=9.8$ $\text{m/s}^2$&Gravitational acceleration\\
				$f_i=0.2$ &Friction coefficient\\
				$\chi_{i1}=0.3 \cdot i$, $\chi_{i2}=0.5 \cdot i$ & Coefficients related to the disturbance \\
				\hline
		\end{tabular}}
	\end{figure}
\end{small}
\begin{figure}[htbp]
		\vspace{-0.2cm}  
	\centering 
	\begin{minipage}[http]{0.36\textwidth} 
		\centering 
		\includegraphics[height=1.9cm,width=5.45cm]{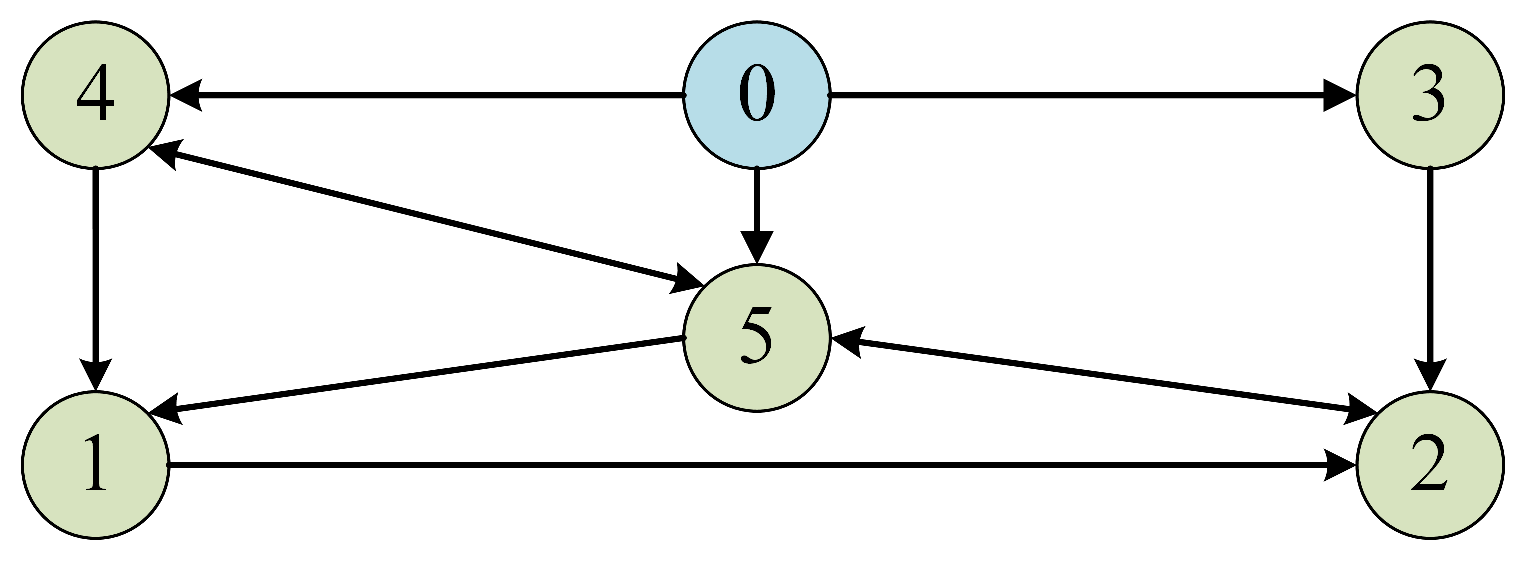}
		\caption{\label{fig1} {Communication graph $\mathcal {G}$}.}
	\end{minipage}
\end{figure}

Choose $ K=1.78\cdot I_5$, which guarantees that $\mathcal H^T K+K \mathcal H-2I_5\geq 0$.
Over the time interval $[0,160]$, the parameters for DoS attacks are selected as follows:  $\nu_d=0.2$ and $p_d=4.9$. 
The design parameters of the observer (\ref{ob}) are specified as $\mu_1=7.5$, $\mu_2=7$, $\mu_3=11$, $\alpha=0.7$, $\beta=1.45$. 
Then, one can calculate that $c_1=21.4662$, $c_2=10.3531$, $c_3=21.8649$, $c_4=10.2899$, $c_5=0.5727$, $\hat c_1= 0.0736$, $\hat c_2=0.1011$, $\tilde c_1=0.0762$, $\tilde c_2=0.1052$, and $t_o=79.5692$ sec. 
It is evident that conditions (i)--(iii) in Theorem \ref{lemmat} are met.
Select $R_{i}=[1, 0, 0, 0]$ with $q_{i1}\!=\!4$, for $i\in \mathcal V_a$.
The control law $u_i(t)$ is designed as in (\ref{controller1n}) with 
$\psi_{i1}^1=2$, $\psi_{i1}^2=4.5$, $\psi_{i1}^3=4.5$, $\psi_{i1}^4=1.8$,
$\bar \psi_{i1}^1=1$, $\bar \psi_{i1}^2=4$, $\bar \psi_{i1}^3=5$, $\bar \psi_{i1}^4=4$,
$\gamma_{i1}^1=0.2727$, $\gamma_{i1}^2=0.3333$, $\gamma_{i1}^3=0.4286$, $\gamma_{i1}^4=0.6$, 
$\bar \gamma_{i1}^1=3$, $\bar \gamma_{i1}^2=2$, $\bar \gamma_{i1}^3=1.5$, $\bar \gamma_{i1}^4=1.2$, 
for $i\in \mathcal V_a$.
Select $Q_{i1}=\bar Q_{i1}=0.02I_4$.
From  (\ref{t_cij}), one has $t_c=69.6789$ sec. In combination with $t_o= 79.5692$ sec, the settling time upper bound for the resulting closed-loop system is obtained as $t_a=149.2480$ sec.

Several simulations  have been conducted over $[0,160]$ under different DoS attacks patterns that satisfy Assumption \ref{assdos} and with various randomly generated initial states for the observers, exosystem and agents within the range of $(-10, 10)$.
Figure \ref{fig2} presents the simulation results for a typical initial state, where $v(0)=[ 1.7135,-2.3265]^T$, 
$x_1(0)=[9.7014,-7.0776,$ $4.2279,-2.6172]^T$,
$x_2(0)=[-4.4514, -3.6204, 3.8463, $ $ -2.5113]^T$, 
$x_3(0)=[3.9321,8.9157,-7.4931,-0.5721]^T$, 
$x_4(0)=$ $[-2.5113,4.0533,4.4166,-4.4445]^T$, 
$x_5(0)=[-3.1023,6.4695,4.3026,4.3101]^T$, and $\Pi_D(0,25)=[0.02,$ $6)\cup[8,9.2)\cup[13,14.5)\cup[25.3,27.4)\cup[39.3,43.2)\cup[62.9,$ $66.4)\cup[77.2,79.3)\cup[83.2,85.5)\cup[113.2,123.5)\cup[153.2,$ $155.5)$.
Figure \ref{fig2} (a) illustrates the fixed-time convergence to zero of the estimation errors $\tilde \eta_i(t)$, for $i\in \mathcal V_a$. Additionally, Figure \ref{fig2} (b) illustrates that the regulated outputs $e_i(t)$, for $i\in \mathcal V_a$ approach zero in a fixed time under the proposed distributed controller, even though the MAS is over directed graphs under DoS attacks.
For comparison, several simulations have been conducted for the distributed resilient fixed-time observer (\ref{ob}) and the distributed resilient exponentially converging observer adapted from \cite{10233067}. The parameters of the fixed-time observer (\ref{ob}) remain the same, except for $\mu_1$.
To be fair, $\mu_1$ is set to the minimum value satisfying each observer's stability condition: $\mu_1=6.6$ for the fixed-time observer (\ref{ob}) and $\mu_1=0.5$ for the exponentially converging observer.
Figure \ref{figD} illustrates the convergence profiles $\|\tilde \eta(t)\|$ for both observers with different initial states. One can observe that the settling times of the exponentially converging observer vary significantly with different initial states while the settling times of the fixed-time observer (\ref{ob}) are almost the same independent of initial states, and furthermore, the convergence times of the fixed-time observer (\ref{ob}) are much shorter than those of the exponentially converging observer.

	\begin{figure}[htbp]
	\vspace{-0.2cm}  
		\centering
		\subfigure[]
		{\includegraphics[height=3.5cm,width=4.2cm]{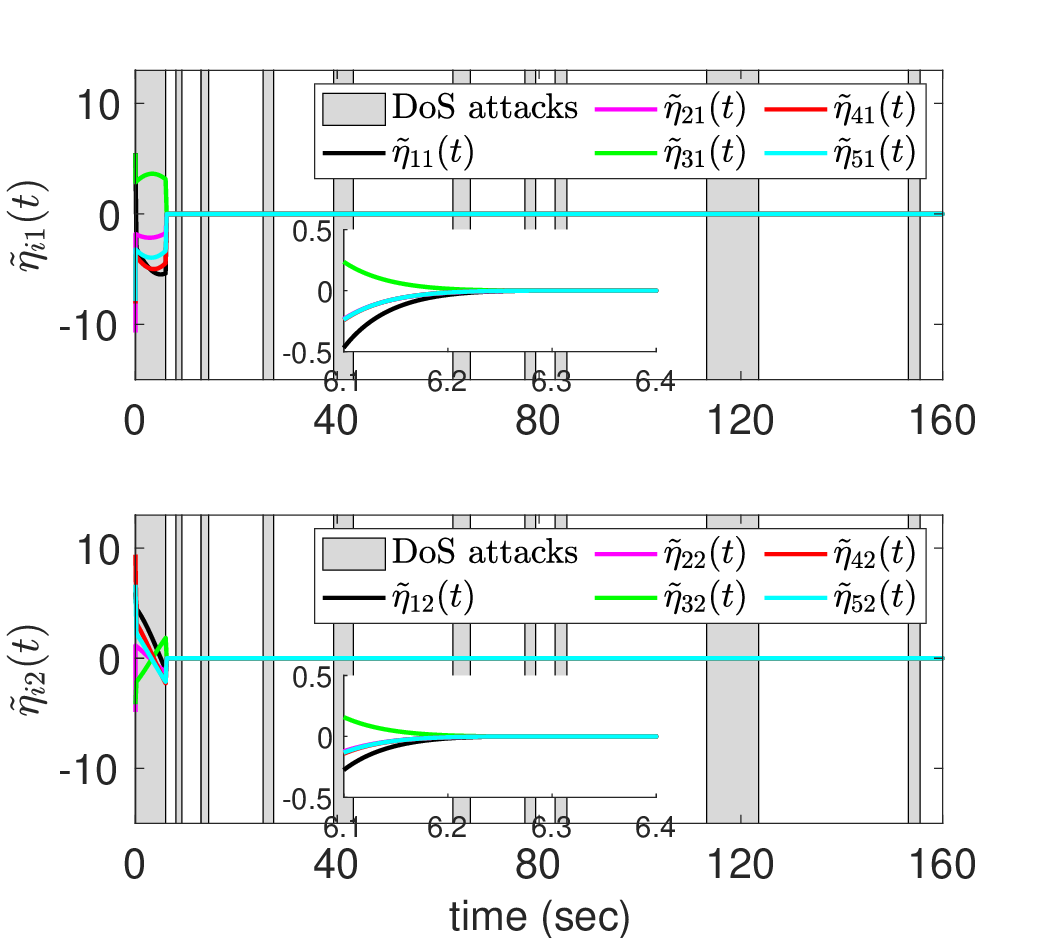}}
		\hspace{-0.5cm}
		\subfigure[]
		{\includegraphics[height=3.5cm,width=4.2cm]{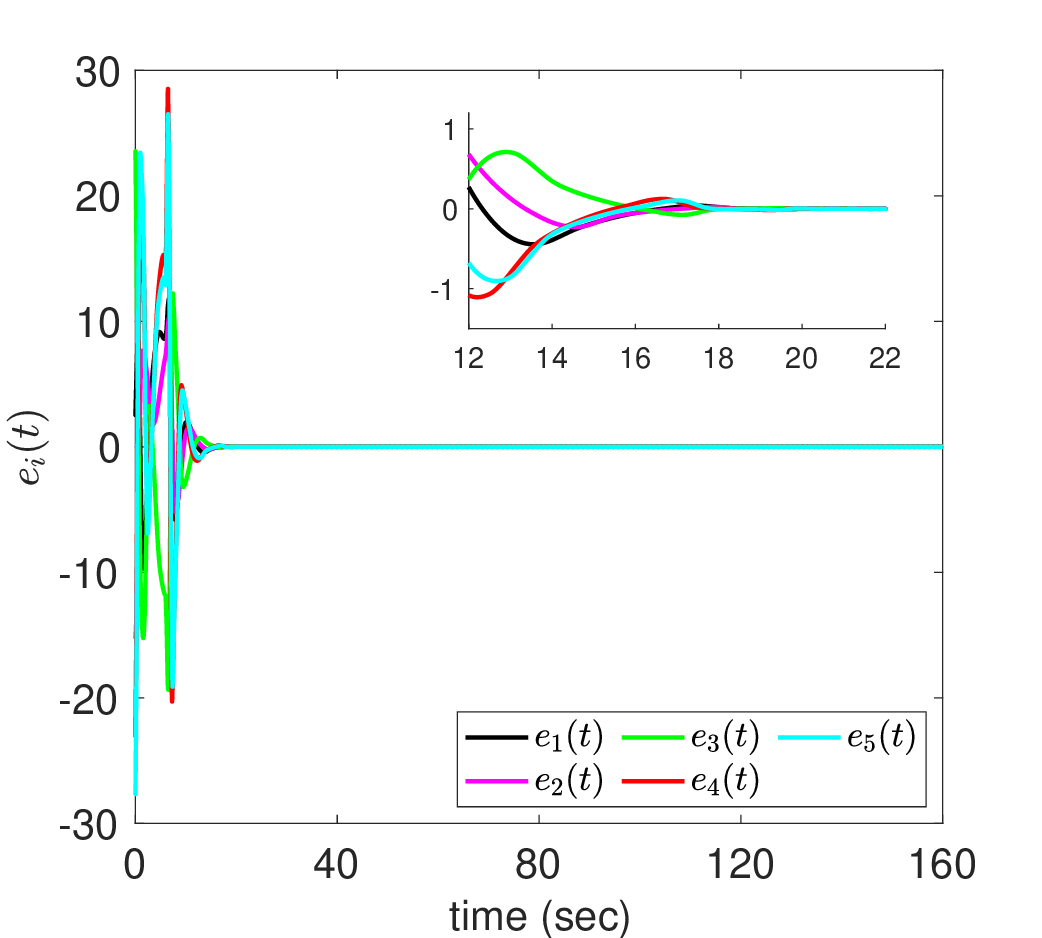}}
	\vspace{-0.3cm} 
	\caption{{Estimation error and regulated output profiles of the MAS under DoS attacks.}} 
	\label{fig2}
\end{figure}	
	\begin{figure}[htbp]
	\vspace{-0.2cm}  
	\centering
		{\includegraphics[height=2.9cm,width=4.1cm]{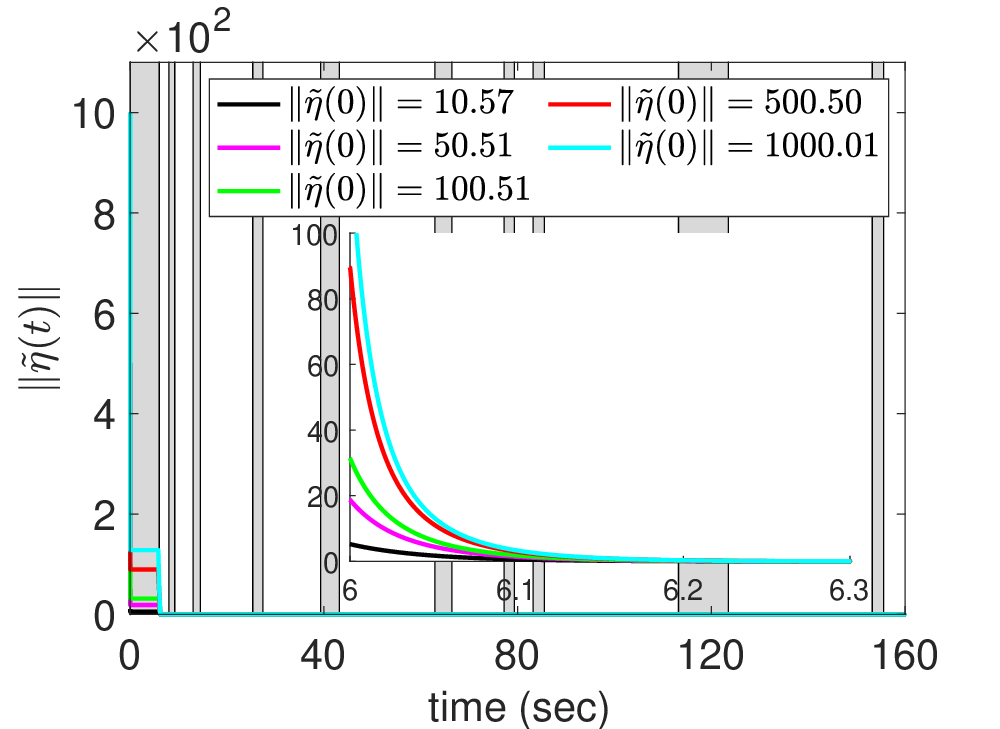}}
		{\includegraphics[height=2.9cm,width=4.1cm]{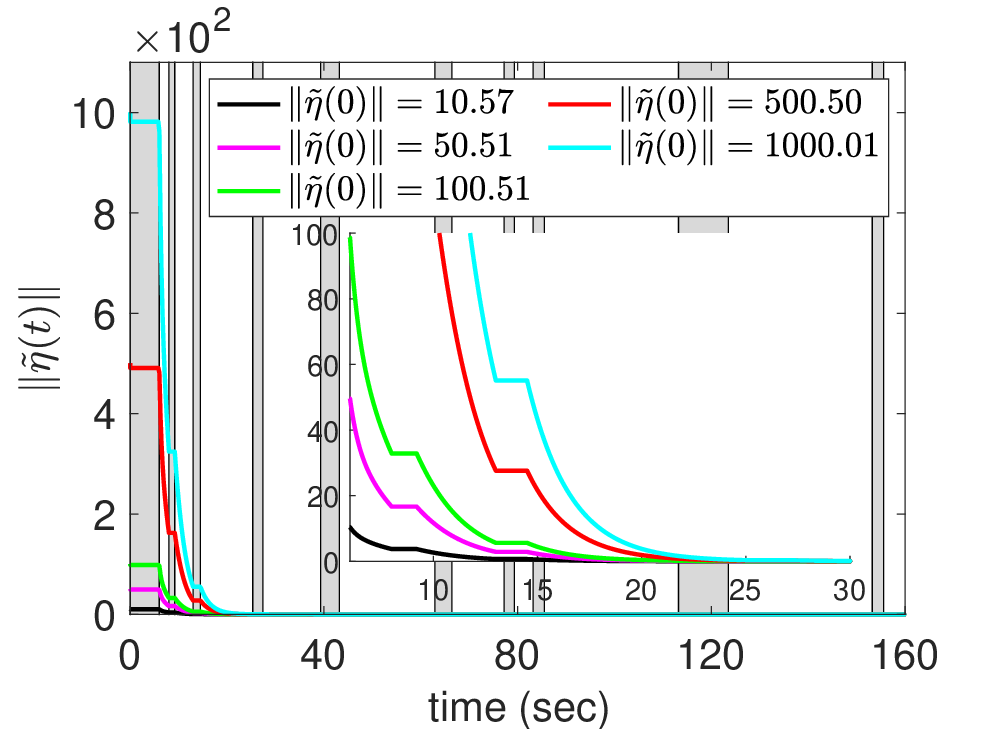}}
	\caption{{Profiles of $\|\tilde \eta(t)\|$ for the fixed-time observer (\ref{ob}) (left) and the exponentially converging observer adapted from \cite{10233067} (right).}} 
\label{figD}
\end{figure}


\begin{figure}[http]
		\centering
		\captionof{table}{{Comparison of convergence performance.}}		
				\scalebox{0.85}{
\begin{tabular}{|c|c|c|c|}
	\hline
	$\alpha$&$\beta$ & $\mu_1$, $\mu_2$, $\mu_3$ & Time (sec) \\
	\hline
	\multirow{3}{*}{0.7} & \multirow{3}{*}{1.45} & 7.5, 7, 11.5 & 6.3057 \\
	\cline{3-4}
	& & 8.5, 9.5, 12.5 & 6.2414 \\
	\cline{3-4}
	& & 9.5, 13, 14.5 & 6.1877 \\
	\hline
\multirow{3}{*}{0.75}& 1.35 & \multirow{3}{*}{9.9, 18.6, 18.6} & 6.1599 \\
	\cline{2-2} \cline{4-4}
& 1.45 & & 6.1583 \\
	\cline{2-2} \cline{4-4}
& 1.55 & & 6.1564 \\
	\hline
0.85& 	\multirow{3}{*}{1.35} & \multirow{3}{*}{8, 11.2, 15} & 6.3111 \\
	\cline{1-1} \cline{4-4}
	 0.8& & & 6.2715\\
	\cline{1-1} \cline{4-4}
	0.75 & & & 6.2391 \\
	\hline
\end{tabular}}
\end{figure}
Furthermore, several simulations are performed to study the impact of the observer design parameters $\alpha$, $\beta$, $\mu_1$, $\mu_2$ and $\mu_3$ on the convergence speed of $\|\tilde \eta(t)\|$. The comparison results are recorded in Table II with an error tolerance threshold of $10^{-4}$. It can be observed from Table II that larger values of $\mu_1$, $\mu_2$, $\mu_3$, $\beta$ and/or a smaller value of $\alpha$ lead to a faster convergence of $\|\tilde\eta(t)\|$.

\section{Conclusions}
This work presents a distributed resilient fixed-time controller for heterogeneous MASs over directed graphs under DoS attacks.
The proposed controller can solve the resilient fixed-time COR problem of the concerned MAS under some mild conditions. Some guidelines for design of the distributed resilient fixed-time controller are also provided.
Future work could explore prescribed-time cooperative control under DoS attacks and extend the fixed-time/prescribed-time framework to more general attack modes instead of the zero-topology attacks.

\begin{appendices}	
\section{Several Technical Lemmas}\label{lemmas}
Several technical lemmas, which are utilized in the proof of the main results of this work, are introduced in this appendix.
\begin{lemma}\label{vector}\citep{mueller2009normal,luenberger1967canonical}
	Consider the linear system 
	\begin{align}\label{linear}
		\dot x(t) & = Ax(t) + Bu(t),
	\end{align}where $x(t) \in \mathbb{R}^n$ and $u(t) \in \mathbb{R}^m$ are its state and input, respectively, $(A,B)$ is controllable and $B$ is of full column rank. 

(i) There is a linear coordinate transformation $y(t)=Rx(t)\in \mathbb R^m$ such that $y(t)$ has complete vector relative degree $(q_1,q_2, \cdots,q_m)$. Moreover, $
\text{det}\left(X\right)\neq 0$, and $R_jA^lB=0$, for $l=0,1, \cdots,q_{j}-2$, $j=1,2,\cdots,m$, where  $X=\left[\begin{matrix}
	\vspace{-0.7em}R_1A^{q_1-1}B\\
	\vdots\\
	R_mA^{q_m-1}B
\end{matrix}	\right]$ and 
$R_j$ is the row indexed by $j$ in $R$. 

(ii) There are nonsingular transformations $\bar u(t)=G {u(t)}$ and $\bar x(t)=T{x(t)}$ such that system (\ref{linear}) is transformed into
	\begin{align}
\dot {\bar x}(t)=\bar A \bar x(t)+\bar B\bar u(t),
	\end{align}where
	$\bar A=[\bar A_{lj}]$, $\bar B=blockdiag\{	\Theta_{q_1},\Theta_{q_2},\cdots, \Theta_{q_m}\}$,
	$\bar A_{{ll}}=
	\left[\begin{matrix}
		\textbf 0_{q_{l}-1}& I_{q_{l}-1}\\
		\zeta_{{ll}}^1&\bar \zeta_{{ll}}\\
	\end{matrix}\right]$, $\bar A_{lj}=\Theta_{q_l}\zeta_{lj}^T,$
	$\bar \zeta_{ll}=[ \zeta_{ll}^2, \cdots, \zeta_{ll}^{q_l}]$,
	$\zeta_{lj}=[\zeta_{lj}^1, $ $\zeta_{lj}^2, \cdots, \zeta_{lj}^{q_j}]$, $l\neq j$, $l,j=1,2,\cdots,m$,
	$\Theta_{q_l}=[0, \cdots, 0, 1]^T\in \mathbb R^{q_l}$, $q_1+q_2+\cdots+q_m=n$. 
\end{lemma}
\begin{lemma}\label{lemmafinite}\citep{basin2016continuous}
Consider the system 
\begin{align}\label{order}
		\dot {\rho}(t)=\left[\begin{matrix}
			\textbf {0}_{n-1}& {I}_{n-1}\\
			0&\textbf 0_{n-1}^T
		\end{matrix}	\right] \rho(t)+\left[\begin{matrix}
		\textbf{0}_{n-1}\\
		1
	\end{matrix}	\right] \omega(t), \rho(0)=\rho_0.
	\end{align}There is $\epsilon\in(0,1)$ such that for $\gamma\in (1-\epsilon,1)$ and $\bar \gamma\in (1,1+\epsilon)$, the origin is globally fixed-time stable under the control law 
$\omega(t)=-\sum_{r=1}^{n}(\psi_r\text{sig}^{\gamma_r}(\rho_r(t))+\bar \psi_r\text{sig}^{\bar \gamma_r}(\rho_r(t))),$
	where $\gamma_r$ and $\bar\gamma_r$ are chosen as follows:
	$\gamma_{r-1}=\tfrac{\gamma_r\gamma_{r+1}}{2\gamma_{r+1}-\gamma_r}$, 
	$\bar \gamma_{r-1}=\tfrac{\bar \gamma_r\bar \gamma_{r+1}}{2\bar \gamma_{r+1}-\bar \gamma_r}$, 
	$r=2, \cdots, n$, $\gamma_{n+1}=\bar \gamma_{n+1}=1$, $\gamma_n=\gamma$, $\bar \gamma_n=\bar\gamma$, and $\psi_r$, $\bar \psi_r$, $r\!=\!1,2,\!\cdots\!,n$,
	are selected such that matrices
	\begin{align*}
		\Psi\!=\left[
		\begin{array}{@{}cccc@{}} 
			\vspace{-0.5em}0 & \hspace{-0.5em} 1 & \hspace{-0.5em} \cdots & \hspace{-0.5em} 0 \\
			\vdots & \hspace{-0.5em} \vdots & \hspace{-0.3em} \ddots & \hspace{-0.5em} \vdots \\
			0 & \hspace{-0.5em} 0 & \hspace{-0.5em} \cdots & \hspace{-0.5em} 1 \\
			-\psi_{1} & \hspace{-0.5em} -\psi_{2} & \hspace{-0.5em} \cdots & \hspace{-0.5em} -\psi_{n}
		\end{array}\right]\!,\quad
		\bar \Psi\!\!=\!\!\left[\!
		\begin{array}{@{}cccc@{}} 
			\vspace{-0.5em}0 & \hspace{-0.5em} 1 & \hspace{-0.5em} \cdots & \hspace{-0.5em} 0 \\
			\vdots & \hspace{-0.5em} \vdots & \hspace{-0.5em} \ddots & \hspace{-0.5em} \vdots \\
			0 & \hspace{-0.5em} 0 & \hspace{-0.5em} \cdots & \hspace{-0.5em} 1 \\
			-\bar \psi_{1} & \hspace{-0.5em} -\bar\psi_{2} & \hspace{-0.5em} \cdots & \hspace{-0.5em} -\bar\psi_{n}
		\end{array}\right]
	\end{align*}are Hurwitz. In addition, the settling time is upper-bounded by 
\begin{equation}
	t_c=\frac{\gamma\lambda_{M}(P)(\lambda_{M}(P))^{\frac{1-\gamma}{\gamma}}}{(1-\gamma)\lambda_{m}(Q)}+\frac{\bar\gamma\lambda_{M}(\bar P)(\lambda_{M}(\bar P))^{\frac{\bar \gamma-1}{\bar\gamma}}}{\lambda_{m}(\bar Q)(\bar \gamma-1)},
	\end{equation}
 where $P>0$ and $\bar P>0$ satisfy the Lyapunov equations:
	$
	P\Psi+\Psi^TP=-Q
	$ and
	$
	\bar P\bar \Psi+\bar \Psi^T\bar P=-\bar Q
	$,
	for any $Q>0$ and $\bar Q>0$, respectively.
\end{lemma}
\begin{lemma}\label{lem2}\citep{hardy1952,zuo2015nonsingular,song2021distributed} 
	For $x_{ij}\geq0$, if $p\in(0,1]$ and $q\in(1,\infty)$, one has
\begin{align*}
	(nm)^{p-1}\sum_{i=1}^{n}\sum_{j=1}^{m}x_{ij}^p\leq 
	(\sum_{i=1}^{n}\sum_{j=1}^{m}x_{ij})^p
	\leq \sum_{i=1}^{n}\sum_{j=1}^{m}x_{ij}^p,\\
\sum_{i=1}^{n}\sum_{j=1}^{m}x_{ij}^q\leq
	(\sum_{i=1}^{n}\sum_{j=1}^{m}x_{ij})^q
	\leq (nm)^{(q-1)}\sum_{i=1}^{n}\sum_{j=1}^{m}x_{ij}^q.
\end{align*}
\end{lemma}

\section{Proof of Theorem \ref{lemmat}}\label{pplemmat}


Define the following Lyapunov function candidate,
\begin{align}\label{v}
	V(t)=&\sum_{i=1}^{N}(\frac{\mu_1k_i}{2}\|\varsigma_i^2(t)\|_1+\frac{\mu_2k_i}{\alpha+1}\|\varsigma_i^{\alpha+1}(t)\|_1\nonumber\\
	&+\frac{\mu_3k_i}{\beta+1}\|\varsigma_i^{\beta+1}(t)\|_1),
\end{align}where $k_i$ is defined in Lemma \ref{lem1} (ii).

By taking DoS attacks into consideration, we will analyze two cases as follows.

Case I: $t\in \Pi_N(t_0,t)$.

From (\ref{obqq11}), the time derivative of $V(t)$ is expressed as follows,
\begin{align*}
		\dot V(t)\!=&
		\sum_{i=1}^{N}k_i(\mu_1\varsigma_i(t)\!+\!\mu_2k_i sig^{\alpha}(\varsigma_i(t))\!+\!\mu_3k_i sig^{\beta}(\varsigma_i(t)))^T\dot \varsigma_i(t)\nonumber\\
		=&\Phi^T(t)(K\otimes S)\varsigma(t)-\Phi^T(t)(K\mathcal H\otimes I_q)\Phi(t)\nonumber\\
		=&\Phi^T(t)(K\otimes S)\varsigma(t)-\frac{1}{2}\Phi^T(t)((\mathcal H^TK+K \mathcal H)\otimes I_q)\Phi(t).
\end{align*}By Young's inequality, one obtains 
	$\Phi^T(t)(K\otimes S)\varsigma(t)
	\le \tfrac{1}{2}\Phi^T(t)\Phi(t)
	+ \tfrac{1}{2}\|(K\otimes S)\|^{2}\varsigma^{T}(t)\varsigma(t)$,
	and Lemma \ref{lem1} (ii) ensures that 
	$\mathcal H^{T}K + K\mathcal H - 2I_{N} \ge 0$. Then, one has
\begin{align}\label{q5}
	\dot V(t)\leq&\frac{1}{2}\|(K\otimes S)\|^2\varsigma^T(t)\varsigma(t)-\frac{1}{2}\Phi^T(t)\Phi(t).
\end{align}
By recalling $\Phi_i(t)=\mu_1 \varsigma_i(t) + \mu_2 \text{sig}^\alpha(\varsigma_i(t)) + \mu_3 \text{sig}^\beta(\varsigma_i(t))$ with $\mu_1$, $\mu_2$, $\mu_3>0$,  one has
\begin{equation}\label{q6}
	\|\Phi_i(t)\|^2
	\geq \mu_1^2\|\varsigma_i(t)\|^2+\mu_2^2\|\text{sig}^\alpha(\varsigma_i(t))\|^2+\mu_3^2\|\text{sig}^\beta(\varsigma_i(t))\|^2.
\end{equation}
Via Lemma \ref{lem2}, one obtains the following two inequalities,
	\begin{align}
		\sum_{i=1}^N\|\text{sig}^\alpha(\varsigma_i(t))\|^2\geq (\|\varsigma(t)\|^2)^\alpha,\label{q7}\\
		\sum_i^N \|\text{sig}^\beta(\varsigma_i(t))\|^2\geq (Nq)^{1-\beta}(\|\varsigma(t)\|^2)^\beta.\label{q8}
	\end{align}
	Combining (\ref{q6}), (\ref{q7}) and (\ref{q8}) leads to
	\begin{align}\label{qp}
		\|\Phi(t)\|^2\geq& \mu_1^2\|\varsigma(t)\|^2+\mu_2^2(\|\varsigma(t)\|^2)^\alpha+\mu_3^2(Nq)^{1-\beta}(\|\varsigma(t)\|^2)^\beta.
\end{align}It follows from (\ref{q5}) and (\ref{qp}) that
\begin{align}\label{q10}
	\dot V(t)\leq&-\frac{\mu_1^2-\|K\otimes S\|}{2}\|\varsigma(t)\|^2-\frac{\mu_2^2}{2}(\|\varsigma(t)\|^2)^\alpha\nonumber\\
	&-\frac{\mu_3^2}{2}(\|\varsigma(t)\|^2)^\beta\nonumber\\
	\leq&-c_1((\|\varsigma(t)\|^2)^\beta+\|\varsigma(t)\|^2+(\|\varsigma(t)\|^2)^\alpha),
\end{align}where $c_1$ is defined below (\ref{to}).

Recall that $0<\alpha<1<\beta$. Applying Lemma~\ref{lem2} to \eqref{v} gives
\begin{align}\label{le1}
	V(t)\leq c_2(\|\varsigma(t)\|^2+(\|\varsigma(t)\|^2)^{\frac{\alpha+1}{2}}+ (\|\varsigma(t)\|^2)^{\frac{\beta+1}{2}}),
\end{align}where $c_2$ is defined below (\ref{to}).
Since $\alpha<\tfrac{\alpha+1}{2}<1$ and $1<\tfrac{\beta+1}{2}<\beta$, one can infer that
\begin{align}
	(\|\varsigma(t)\|^2)^{\frac{\alpha+1}{2}}\leq \|\varsigma(t)\|^2+(\|\varsigma(t)\|^2)^\alpha,\label{le2}\\
	(\|\varsigma(t)\|^2)^{\frac{\beta+1}{2}}\leq \|\varsigma(t)\|^2+(\|\varsigma(t)\|^2)^\beta.\label{le2.1}
\end{align}
Then, substituting (\ref{le2}) and (\ref{le2.1}) into (\ref{le1}) leads to
\begin{align}\label{Vc2}
	V(t)\leq c_2((\|\varsigma(t)\|^2)^\beta+ (\|\varsigma(t)\|^2)^\alpha+3\|\varsigma(t)\|^2).
\end{align}

Since $\tfrac{1}{\alpha}<\beta$, one has $\tfrac{(\alpha+1)\beta}{\beta+1}>1$ and $\tfrac{(\beta+1)\alpha}{\alpha+1}>1$. Then, applying Lemma~\ref{lem2} to \eqref{v} along with $\tfrac{2\beta}{\beta+1}>1$, $\beta>1$, $0<\tfrac{2\alpha}{\alpha+1}<1$ and $0<\alpha<1$, yields
\begin{eqnarray}
	V^{\frac{2\beta}{\beta+1}}\!(t)\!\leq\!
c_3((\|\varsigma(t)\|^2)^{\frac{2\beta}{\beta+1}}\!\!+\!(\|\varsigma(t)\|^2)^{\frac{\beta(\alpha+1)}{\beta+1}}\!\!+\!(\|\varsigma(t)\|^2)^\beta),\!\label{le3}\\
		V^{\frac{2\alpha}{\alpha+1}}\!(t)\!\leq\! c_4 ((\|\varsigma(t)\|^2)^{\frac{2\alpha}{\alpha+1}}\!+\!(\|\varsigma(t)\|^2)^{\alpha}\!+\!(\|\varsigma(t)\|^2)^{\frac{(\beta+1)\alpha}{\alpha+1}}),\!\label{le4}
\end{eqnarray}
where $c_3$ and $c_4$ are defined below (\ref{to}).
By $1<\tfrac{2\beta}{\beta+1}<\beta$, $\alpha<$ $\tfrac{\beta(\alpha+1)}{\beta+1}< \beta$, $\alpha<\tfrac{2\alpha}{ \alpha+1}<1$ and $\alpha<\tfrac{(\beta+1)\alpha}{\alpha+1}<\beta$, 
one has 
\begin{align}
	(\|\varsigma(t)\|^2)^{\frac{2\beta}{\beta+1}}\leq& \|\varsigma(t)\|^2+(\|\varsigma(t)\|^2)^\beta,\label{le5}\\
	(\|\varsigma(t)\|^2)^{\frac{\beta(\alpha+1)}{\beta+1}}\leq& (\|\varsigma(t)\|^2)^\alpha+(\|\varsigma(t)\|^2)^\beta,\label{le6}\\
	(\|\varsigma(t)\|^2)^{\frac{2\alpha}{\alpha+1}}\leq& \|\varsigma(t)\|^2+(\|\varsigma(t)\|^2)^\alpha,\label{le7}\\
	(\|\varsigma(t)\|^2)^{\frac{(\beta+1)\alpha}{\alpha+1}}\leq& (\|\varsigma(t)\|^2)^\alpha+(\|\varsigma(t)\|^2)^\beta.\label{le8}
\end{align}
Substituting (\ref{le5}) and (\ref{le6}) into (\ref{le3}), and substituting (\ref{le7}) and (\ref{le8}) into (\ref{le4}) yield respectively,
\begin{align}\label{q13}
	V^{\frac{2\beta}{\beta+1}}(t)\leq& c_3(3(\|\varsigma(t)\|^2)^\beta+\|\varsigma(t)\|^2+(\|\varsigma(t)\|^2)^{\alpha}),
\end{align}
\begin{align}\label{q16}
	V^{\frac{2\alpha}{\alpha+1}}(t)\leq&c_4 ((\|\varsigma(t)\|^2)^{\beta}+3(\|\varsigma(t)\|^2)^{\alpha}+\|\varsigma(t)\|^2).
\end{align}

According to (\ref{Vc2}), (\ref{q13}) and (\ref{q16}), one has
\begin{align}\label{q17}
	&	\frac{1}{c_2}V(t)+\frac{1}{c_3}V^{\frac{2\beta}{\beta+1}}(t)+\frac{1}{c_4}V^{\frac{2\alpha}{\alpha+1}}(t)
\nonumber\\
	&\leq 5((\|\varsigma(t)\|^2)^{\alpha}+\|\varsigma(t)\|^2+(\|\varsigma(t)\|^2)^{\beta}).
\end{align}
Substituting (\ref{q17}) into (\ref{q10}) leads to
\begin{align}\label{q18}
	\dot V(t)\leq&-\frac{c_1}{5c_2}V(t)-\frac{c_1}{5c_3}V^{\frac{2\beta}{\beta+1}}(t)-\frac{c_1}{5c_4}V^{\frac{2\alpha}{\alpha+1}}(t).
\end{align}

Case II: $t\in \Pi_D(t_0,t)$.

Taking the time derivative of $V(t)$ along (\ref{obqq12}) yields
\begin{align*}
	\dot V(t)	
	\!=&\sum_{i=1}^{N}k_i(\mu_1\varsigma_i(t)\!+\!\mu_2\text{sig}^\alpha(\varsigma_i(t))\!+\!\mu_3\text{sig}^\beta(\varsigma_i(t)))^TS\varsigma_i(t)\nonumber\\
	\leq&|\!\sum_{i=1}^{N}\!k_i(\mu_1\varsigma_i(t)\!+\!\mu_2\text{sig}^\alpha(\varsigma_i(t))\!+\!\mu_3\text{sig}^\beta(\varsigma_i(t)))^TS\varsigma_i(t)|.
\end{align*}
Then, by the triangle inequality, one further has
\begin{align}\label{q19}
		\dot V(t)	\leq&\sum_{i=1}^{N}k_i(\mu_1|\varsigma_i^T(t)S\varsigma_i(t)|+\mu_2|(\text{sig}^{\alpha}(\varsigma_i(t)))^TS\varsigma_i(t)|\nonumber\\
	&+\mu_3|(\text{sig}^{\beta}(\varsigma_i(t)))^TS\varsigma_i(t)|).
\end{align}Note that
\begin{align}\label{q20}
	|\varsigma_i^T(t)S\varsigma_i(t)|\leq \|S\|\|\varsigma_i^2(t)\|_1.
\end{align}
In addition, via Lemma \ref{lem2}, one can infer that
\begin{align}
	|(\text{sig}^{\alpha}(\varsigma_i(t)))^TS\varsigma_i(t)|\leq q^{\frac{1-\alpha}{2}}\|S\|\|\varsigma_i^{\alpha+1}(t)\|_1,\label{q21}\\
	|(\text{sig}^{\beta}(\varsigma_i(t)))^TS\varsigma_i(t)|\leq q^{\frac{\beta-1}{2}}\|S\|\|\varsigma_i^{\beta+1}(t)\|_1.\label{q22}
\end{align}
Substituting (\ref{q20})--(\ref{q22}) into (\ref{q19}) leads to
\begin{align}\label{q23}
	\dot V(t)\leq&\|S\|(\sum_{i=1}^{N}k_i (\mu_1\|\varsigma_i^2(t)\|_1+q^{\frac{1-\alpha}{2}}\mu_2\|\varsigma_i^{\alpha+1}(t)\|_1\nonumber\\
		&+q^{\frac{\beta-1}{2}}\mu_3\|\varsigma_i^{\beta+1}(t)\|_1))\nonumber\\
	\leq&c_5V(t),
\end{align}where $c_5$ is defined below (\ref{to}).

As stated in Definition \ref{FTstable} (2), to show the fixed-time stability of the origin $\varsigma(t)=0$, i.e., the equilibrium of system (\ref{obq1}), 
it is necessary to first demonstrate its global finite-time stability. 
According to Definition \ref{FTstable} (1), the proof of global finite-time stability comprises two parts: Part I establishes global finite-time convergence, and Part II shows Lyapunov stability.

Part I: The global finite-time convergence is proven by the following two scenarios.

Scenario 1: $V(t_0)>1$.

The convergence property of $V(t)$ will be demonstrated by two steps in this scenario. In Step 1, we show that there is a $\bar t_o$ for which $V(t) \le 1$ holds, for all $t \ge \bar t_o$. Subsequently, in Step 2, we establish the existence of $t_o$ such that $V(t) = 0$, for all $t \ge t_o$.

Step 1: By invoking (\ref{q18}), one has 
\begin{align}\label{obdV}
	\dot V(t)\leq
	-\frac{c_1}{5c_2}V(t)-\frac{c_1}{5c_3}V^{\frac{2\beta}{\beta+1}}(t),
	t\in [t_{k-1}, t^a_k).
\end{align}
Then, for $t\in [t_{k-1}, t^a_k)$, one can deduce from (\ref{obdV}) that 
\begin{align}\label{c1}
	V^{\frac{1-\beta}{\beta+1}}(t)\geq (V^{\frac{1-\beta}{\beta+1}}(t_{k-1})+\frac{c_2}{c_3})e^{\tilde c_1(t-t_{k-1})}-\frac{c_2}{c_3},
\end{align}where $\tilde c_1=\tfrac{c_1(\beta-1)}{5c_2(\beta+1)}$.

For $t\in [t^a_k,t_k)$, one can deduce from (\ref{q23}) that
\begin{align}\label{c2}
	V^{\frac{1-\beta}{\beta+1}}(t)\geq V^{\frac{1-\beta}{\beta+1}}(t_k^a)e^{-\tilde c_2(t-t_k^a)},
\end{align}where $\tilde c_2=\tfrac{c_5(\beta-1)}{\beta+1}$.

From (\ref{c1}) and (\ref{c2}), one has for $t\in [t_{k-1}, t^a_k)$,
\begin{align}\label{obdW2}
		V^{\frac{1-\beta}{\beta+1}}(t)
		\geq&(V^{\frac{1-\beta}{\beta+1}}(t_{k-1})+\frac{c_2}{c_3})e^{\tilde c_1(t-t_{k-1})}-\frac{c_2}{c_3}\nonumber\\
		\geq& V^{\frac{1-\beta}{\beta+1}}(t_{k-1}^a)e^{-\tilde c_2(t_{k-1}-t_{k-1}^a)+\tilde c_1(t-t_{k-1})}\nonumber\\
		&+\frac{c_2}{c_3}e^{\tilde c_1(t-t_{k-1})}-\frac{c_2}{c_3}\nonumber\\
		\geq&(V^{\frac{1-\beta}{\beta+1}}(t_{k-2})+\frac{c_2}{c_3})e^{-\tilde c_2(t_{k-1}-t_{k-1}^a)}\nonumber\\
		&\times e^{\tilde c_1((t-t_{k-1})+(t_{k-1}^a-t_{k-2}))}\nonumber\\
		&-\frac{c_2}{c_3}e^{-\tilde c_2(t_{k-1}-t_{k-1}^a)+\tilde c_1(t-t_{k-1})}\nonumber\\
		&+\frac{c_2}{c_3}e^{\tilde c_1(t-t_{k-1})}-\frac{c_2}{c_3}\nonumber\\
		\geq& \cdots\nonumber\\
		\geq& (V^{\frac{1-\beta}{\beta+1}}(t_0)+\frac{c_2}{c_3})e^{-\tilde c_2(\sum_{l=1}^{k-1}(t_l-t_l^a))}\nonumber\\
		&\times e^{\tilde c_1((t-t_{k-1})+\sum_{l=1}^{k-1}(t_{l}^a-t_{l-1}))}-\frac{c_2}{c_3};
\end{align}
and for $t\in [ t^a_k,t_k)$,
\begin{align}\label{obdW3}
		V^{\frac{1-\beta}{\beta+1}}(t)
		\geq& V^{\frac{1-\beta}{\beta+1}}(t_k^a)e^{-\tilde c_2(t-t_k^a)}\nonumber\\
		\geq& (V^{\frac{1-\beta}{\beta+1}}(t_{k-1})+\frac{c_2}{c_3})e^{-\tilde c_2(t-t_k^a)+\tilde c_1(t_k^a-t_{k-1})}\nonumber\\
		&-\frac{c_2}{c_3}e^{-\tilde c_2(t-t_k^a)}\nonumber\\
		\geq& V^{\frac{1-\beta}{\beta+1}}(t_{k-1}^a)e^{-\tilde c_2((t-t_k^a)+(t_{k-1}-t_{k-1}^a))+\tilde c_1(t_k^a-t_{k-1})}\nonumber\\
		&+\frac{c_2}{c_3}e^{-\tilde c_2(t-t_k^a)+\tilde c_1(t_k^a-t_{k-1})}-\frac{c_2}{c_3}e^{-\tilde c_2(t-t_k^a)}\nonumber\\
		\geq& (V^{\frac{1-\beta}{\beta+1}}(t_{k-2})+\frac{c_2}{c_3})e^{-\tilde c_2((t-t_k^a)+(t_{k-1}-t_{k-1}^a))}\nonumber\\
		&\times e^{\tilde c_1((t_k^a-t_{k-1})+(t_{k-1}^a-t_{k-2}))}\nonumber\\
		&-\frac{c_2}{c_3}e^{-\tilde c_2((t-t_k^a)+(t_{k-1}-t_{k-1}^a))+\tilde c_1(t_k^a-t_{k-1})}\nonumber\\
		&+\frac{c_2}{c_3}e^{-\tilde c_2(t-t_k^a)+\tilde c_1(t_k^a-t_{k-1})}-\frac{c_2}{c_3}e^{-\tilde c_2(t-t_k^a)}\nonumber\\
		\geq&\cdots\nonumber\\
		\geq&(V^{\frac{1-\beta}{\beta+1}}(t_0)+\frac{c_2}{c_3})e^{-\tilde c_2((t-t_k^a)+\sum_{l=1}^{k-1}(t_l-t_l^a))}\nonumber\\
		&\times e^{\tilde c_1(\sum_{l=1}^k(t_l^a-t_{l-1}))}-\frac{c_2}{c_3}e^{-\tilde c_2(t-t_k^a)}.
\end{align}
Under Assumption \ref{assdos}, combining (\ref{obdW2}) and (\ref{obdW3}) leads to
\begin{align}\label{obdW5}
	V^{\frac{1-\beta}{\beta+1}}(t)\geq&\frac{c_2}{c_3}e^{(\tilde c_1(p_d-1)-\tilde c_2)\frac{(t-t_0)}{p_d}-\tilde c_1\nu_d-\tilde c_2\nu_d}-\frac{c_2}{c_3}.
\end{align}
Denote 
\begin{align}
	g_1(t)=\frac{c_2}{c_3}e^{(\tilde c_1(p_d-1)-\tilde c_2)\frac{(t-t_0)}{p_d}-\tilde c_1\nu_d-\tilde c_2\nu_d}-\frac{c_2}{c_3}-1.
\end{align}
Then, one has $g_1(t_0)<0$ and 
	\begin{align}
		\dot g_1(t)=\frac{c_2(\tilde c_1(p_d-1)-\tilde c_2)}{c_3p_d} e^{(\tilde c_1(p_d-1)-\tilde c_2)\frac{(t-t_0)}{p_d}-\tilde c_1\nu_d-\tilde c_2\nu_d}.
\end{align}
Together with condition (i), one can infer that there is $\bar t_o$ such that $g_1(t)\geq 0$, for all $t\geq \bar t_o$, in which $\bar t_o$ is determined by $g_1(t)=0$.
Thus, $V^{\frac{1-\beta}{\beta+1}}(t)\geq 1$, for all $t\geq \bar t_o$. Since $\tfrac{1-\beta}{\beta+1}<0$, it can be deduced that $V(t)\leq 1$, for all $t\geq \bar t_o$.

Step 2: By invoking (\ref{q18}), one has for $t\in [t_{k-1}, t^a_k)$,
\begin{align}\label{obdV1}\
	\dot V(t)\leq-\frac{c_1}{5c_4}V^{\frac{2\alpha}{\alpha+1}}(t),
\end{align}
which leads to
\begin{align}\label{obdV2}
	V^{\frac{1-\alpha}{\alpha+1}}(t)\leq
	V^{\frac{1-\alpha}{\alpha+1}}(t_{k-1})-\hat c_1(t-t_{k-1}), 
\end{align}
where $\hat c_1=\tfrac{c_1(1-\alpha)}{5c_4(\alpha+1)}$.

From (\ref{q23}), one has for $t\in [t_k^a,t_{k})$,
\begin{align}\label{obndV2}
	V^{\frac{1-\alpha}{\alpha+1}}(t)\leq V^{\frac{1-\alpha}{\alpha+1}}(t_{k}^a)e^{\hat c_2(t-t_k^a)}, 
\end{align}where $\hat c_2=\tfrac{c_5(1-\alpha)}{\alpha+1}$.

Without loss of generality, denote $\bar t_0=\bar t_o$,  and define $[\bar t_{k-1}, \bar t_k^a)$ as the $k$th interval without DoS attacks, and $[\bar t_k^a,\bar t_k)$ as the $k$th interval with DoS attacks.
From (\ref{obdV2}) and (\ref{obndV2}), one has for $[\bar t_{k-1}, \bar t_k^a)$,
\begin{align}\label{W3}
		V^{\frac{1-\alpha}{\alpha+1}}(t)\leq& V^{\frac{1-\alpha}{\alpha+1}}(\bar t_{k-1})-\hat c_1(t-\bar t_{k-1})\nonumber\\
		\leq&V^{\frac{1-\alpha}{\alpha+1}}(\bar t_{k-1}^a)e^{\hat c_2(\bar t_{k-1}-\bar t_{k-1}^a)}-\hat c_1(t-\bar t_{k-1})\nonumber\\
		\leq&V^{\frac{1-\alpha}{\alpha+1}}(\bar t_{k-2})e^{\hat c_2(\bar t_{k-1}-\bar t_{k-1}^a)}\nonumber\\
		&-\hat c_1(\bar t_{k-1}^a-\bar t_{k-2})e^{\hat c_2(\bar t_{k-1}-\bar t_{k-1}^a)}-\hat c_1(t-\bar t_{k-1})\nonumber\\
		\leq&\cdots\nonumber\\
		\leq&V^{\frac{1-\alpha}{\alpha+1}}(\bar t_0)e^{\hat c_2(\sum_{l=1}^{k-1}( \bar t_l-\bar t_l^a))}\nonumber\\
		&-\hat c_1((t-\bar t_{k-1})+\sum_{l=1}^{k-1}(\bar t_{l}^a-\bar t_{l-1}));
\end{align}
and for $[\bar t_k^a,\bar t_{k})$,
\begin{align}\label{W4}
		V^{\frac{1-\alpha}{\alpha+1}}(t)\leq&V^{\frac{1-\alpha}{\alpha+1}}(\bar t_{k}^a)e^{\hat c_2(t-\bar t_k^a)}\nonumber\\
		\leq &V^{\frac{1-\alpha}{\alpha+1}}(\bar t_{k-1})e^{\hat c_2(t-\bar t_k^a)}-\hat c_1(\bar t_{k}^a-\bar t_{k-1})e^{\hat c_2(t-\bar t_k^a)}\nonumber\\
		\leq&V^{\frac{1-\alpha}{\alpha+1}}(\bar t_{k-1}^a)e^{\hat c_2((t-\bar t_k^a)+(\bar t_{k-1}-\bar t_{k-1}^a))}\nonumber\\
		&-\hat c_1(\bar t_{k}^a-\bar t_{k-1})e^{\hat c_2(t-\bar t_k^a)}\nonumber\\
		\leq&V^{\frac{1-\alpha}{\alpha+1}}(\bar t_{k-2})e^{\hat c_2((t-\bar t_k^a)+(\bar t_{k-1}-\bar t_{k-1}^a))}\nonumber\\
		&-\hat c_1(\bar t_{k-1}^a-\bar t_{k-2})e^{\hat c_2((t-\bar t_k^a)+(\bar t_{k-1}-\bar t_{k-1}^a))}\nonumber\\
		&-\hat c_1(\bar t_{k}^a-\bar t_{k-1})e^{\hat c_2(t-\bar t_k^a)}\nonumber\\
		\leq&\cdots\nonumber\\
		\leq&V^{\frac{1-\alpha}{\alpha+1}}(\bar t_{0})e^{\hat c_2((t-\bar t_k^a)+\sum_{l=1}^{k-1}(\bar t_{l}-\bar t_{l}^a))}\nonumber\\
		&-\hat c_1(\sum_{l=1}^{k}(\bar t_{l}^a-\bar t_{l-1})).
\end{align}
Since $V^{\frac{1-\alpha}{\alpha+1}}(\bar t_{0})\leq 1$,
combining (\ref{W3}) and (\ref{W4}) yields
\begin{align}\label{W5}
	V^{\frac{1-\alpha}{\alpha+1}}(t)\leq&e^{\hat c_2|\Pi_D(\bar t_0,t)|}-\hat c_1(t-\bar t_0-|\Pi_D(\bar t_0,t)|).
\end{align}
Under Assumption \ref{assdos},  one can further deduce that
\begin{align}\label{W6}
	V^{\frac{1-\alpha}{\alpha+1}}(t)\leq g_2(t),
\end{align}
where 
\begin{align}\label{g2}
	g_2(t)=e^{\hat c_2(\frac{t-\bar t_0}{p_d}+\nu_d)}-\frac{\hat c_1(p_d-1)}{p_d}(t-\bar t_0)+\hat c_1\nu_d.
\end{align} 
Then, one has
	\begin{align}
		g_2(\bar t_0)&=e^{\hat c_2\nu_d}+\hat c_1\nu_d,\\
		\dot g_2(t)&=\frac{\hat c_2}{p_d}e^{\hat c_2(\frac{t-\bar t_0}{p_d}+\nu_d)}-\frac{\hat c_1(p_d-1)}{p_d},\\
		\ddot g_2(t)&=(\frac{\hat c_2}{p_d})^2e^{\hat c_2(\frac{t-\bar t_0}{p_d}+\nu_d)}.
\end{align}
Note that $	g_2(\bar t_0)>0$. It follows from conditions (ii)-(iii)
that $g_2(t_o)\leq 0$ is the minimum value of $g_2(t)$, which implies that $V^{\frac{1-\alpha}{\alpha+1}}(t_o)\leq 0$. Together with $V(t)\geq 0$, $\forall t\geq 0$, one can infer that $V(t_o)=0$. The finite-time convergence of $V(t)$ is established for this scenario.

Scenario 2: $V(t_0)\leq 1$.

The global finite-convergence of $V(t)$ for this scenario follows directly from the second step of the previous scenario.

Combining the results in these two scenarios, one can conclude that $\lim\limits_{t\to t_o}\varsigma(t)=0$, $\forall \varsigma(t_0)\in \mathbb R^q$. Thus, the origin $\varsigma(t)=0$ is globally finite-time convergent.

Part II: To show the Lyapunov stability.

Combining (\ref{q18}) and (\ref{q23}), one has 
\begin{align}
	\dot V(t)\leq c_5V(t), \forall t\geq t_0.
	\end{align}
From this inequality, one has for all $t\in[t_0, t_o)$,
\begin{align}
	V(t)\leq e^{c_5(t_o-t_0)}V(t_0).
\end{align}
This implies that $V(t)$ remains bounded over $[t_0, t_o)$. 
Together with the result derived in Part I that $V(t_o)=0$,
one can infer that for each open neighborhood $\mathcal{W}_\varsigma$ of the origin, there is an open neighborhood $\mathcal{W}_0 \subset \mathbb R^n$ of the origin such that for any $\varsigma(t_0) \in \mathcal{W}_0 \setminus \{0\}$, $\varsigma(t) \in \mathcal{W}_\varsigma$ holds, for all $t \in [0, t_o)$. The Lyapunov stability of $\varsigma(t)=0$ is thus obtained in the sense of Definition \ref{finiteq} (1) (b).

According to the results derived in Part I and Part II, along with Definition \ref{finiteq} (1), one can conclude that the origin $\varsigma(t)=0$ is globally finite-time stable.
Note that $t_o$ is independent of the initial states. Therefore, the origin $\varsigma(t)=0$ is globally fixed-time stable. 

Recall that $\varsigma(t)=(\mathcal H\otimes I_q)\tilde \eta(t)$ and $(\mathcal H\otimes I_q)$ is invertible from Lemma \ref{lem1} (i). Consequently, we conclude that $\tilde \eta_i(t)=0$, for all $t\geq t_o$. This completes the proof.
\end{appendices}

\renewcommand\refname{References}
\footnotesize
\scriptsize
\bibliographystyle{IEEEtran}
\bibliography{IEEEabrv,referencesuncertainleader}

%

\end{document}